\documentclass[12pt]{article}
\usepackage[margin=1in]{geometry}
\usepackage[english]{babel}
\usepackage{cancel}
\usepackage[colorlinks=true, allcolors=blue]{hyperref}
\usepackage{amsmath}
\usepackage{amsthm}
\usepackage{amssymb}
\usepackage{bbm}
\usepackage{bm}
\usepackage[dvipsnames]{xcolor}

\usepackage{rotating}
\usepackage{multirow}
\usepackage{graphicx}
\usepackage{threeparttable}
\usepackage{subcaption}
\usepackage{natbib}
\usepackage{booktabs}
\usepackage{array}

\newcommand{\ty}{\tilde{y}}

\newcommand{\bty}{\boldsymbol{\tilde{y}}}
\newcommand{\by}{\boldsymbol{y}}
\newcommand{\bx}{\boldsymbol{x}}
\newcommand{\bu}{\boldsymbol{u}}

\newcommand{\bbeta}{\boldsymbol{\beta}}
\newcommand{\bxi}{\boldsymbol{\xi}}
\newcommand{\YY}{\mathcal{Y}}
\newcommand{\wbxi}{\widetilde{\boldsymbol{\xi}}}
\newcommand{\lama}{\lambda^a}
\newcommand{\lams}{\lambda^\star}
\newcommand{\lamsa}{\lambda^{\star a}}

\newtheorem{proposition}{Proposition}

{
    \title{\bf A Bayesian Copula Dirichlet Process GLM for Longitudinal Data}
    \author{
Shijie Yuan\thanks{Department of Statistics and Data Sciences, The University of Texas at Austin, Austin, TX 78712, USA; Email: shijie2024@outlook.com}\ ,
Entejar Alam\thanks{Department of Statistics and Data Sciences, The University of Texas at Austin, Austin, TX 78712, USA; Email: entejar@utexas.edu}\ ,  
Asmit Chakraborty\thanks{Department of Statistics and Data Sciences, The University of Texas at Austin, Austin, TX 78712, USA; Email: asmit.chakraborty@utexas.edu}\ ,  
Paul J. Rathouz\thanks{Department of Population Health, The University of Texas at Austin, Austin, TX 78712, USA; Email: paul.rathouz@austin.utexas.edu}\ , \\
Jonathan S. Schildcrout\thanks{Department of Biostatistics, Vanderbilt University, Nashville, TN 37203, USA; Email: jonny.schild@vumc.org}\ ,
and Peter M\"uller\thanks{Department of Mathematics, The University of Texas at Austin, Austin, TX 78712, USA; Corresponding email: pmueller@math.utexas.edu}\     
    }
\date{}
}

\begin{document}
\maketitle
\begin{abstract}
We introduce a semiparametric Bayesian generalized linear model (GLM) for
longitudinal continuous outcomes that preserves the usual interpretation of
the GLM regression structure while leaving the baseline distribution
nonparametric. The proposed copula Dirichlet process generalized linear model
(CDPGLM) combines a Dirichlet process prior for the baseline distribution, a
parametric GLM regression structure for covariate effects, and a copula for
longitudinal within--subject dependence. This construction separates modeling 
the marginal distribution of an individual repeated measurement from 
modeling the association: the Dirichlet process GLM component defines the 
marginal density and regression mean, whereas the copula 
captures dependence among repeated measurements. In simulations, we compare 
the CDPGLM with a beta mixed-effect model with a subject-specific random intercept and two 
generalized estimating equation (GEE) approaches using exchangeable and 
independent working correlation structures. The methods show broadly 
comparable recovery of regression effects, while the CDPGLM becomes more 
accurate for estimating marginal distributions and lower-tail probabilities 
as the sample size increases. In an application to data from the National 
Heart, Lung, and Blood Institute (NHLBI) Lung Health Study, the CDPGLM and 
the beta mixed model showed qualitatively similar regression patterns, 
but the CDPGLM more closely matched the empirical marginal distribution 
and reported more realistic and detailed lower--tail risk summaries.  
\end{abstract}


\section{Introduction}

We introduce a semiparametric Bayesian generalized linear model (GLM) for longitudinal
data (repeated measurements) in a way that retains the interpretation in the GLM 
regression structure while remaining nonparametric on the baseline distribution. The 
construction combines a copula for the multivariate dependence structure, a Dirichlet
process (DP) prior on the baseline distribution, and the parametric GLM regression structure.  
GLMs provide a unified framework for relating the covariate-specific mean of the 
outcome to a linear predictor by way of a link function while
assuming a parametric outcome distribution
\citep{nelder1972generalized,mccullagh1989generalized}. Although widely used,
classical GLMs rely on distributional assumptions that may be difficult
to justify in biomedical and other applied settings, particularly when outcomes
exhibit skewness, heavy tails, multimodality arising from population heterogeneity, 
or when inference focuses on tail probabilities, as is often the case in biomedical 
studies. To relax these assumptions while retaining the regression structure of GLMs,
\citet{rathouz2009generalized} introduced the semiparametric generalized linear
model (SPGLM), also known as the generalized linear density ratio model (GLDRM) 
\citep{anderson1972separate}. In this framework, the outcome distribution is 
obtained by exponentially tilting an unspecified baseline distribution, thereby 
retaining the GLM regression structure without fully specifying the outcome 
distribution. Subsequent theoretical and methodological
developments have extended the model to outcomes with infinite support and
established existence, identifiability, and consistency under mild regularity
conditions \citep{huang2014joint}, along with efficient likelihood--based
estimation procedures and software implementation in the \texttt{gldrm} R
package \citep{wurm2018semiparametric}.

While classical GLMs have long been implemented within a Bayesian framework
through priors on regression parameters and parametric outcome families
\citep{dey2000generalized}, one of the first Bayesian formulations of semiparametric
regression placed Dirichlet priors on discretized baseline distributions,
yielding the Dirichlet GLMs (Dir--GLMs) suitable for a finite discrete space
\citep{alam2025dir}. However, such GLM constructions do not naturally extend to
continuous support. \citet{alam2025dpglm} addressed this limitation by modeling the
baseline measure as a random measure on continuous support, leading to the
Dirichlet process generalized linear model (DPGLM).
Under the DPGLM, the outcome distribution at each covariate profile
is generated through exponential tilting of a shared random baseline measure, 
preserving the SPGLM regression structure while enabling full Bayesian inference 
and principled uncertainty quantification.

While SPGLMs and their Bayesian extensions substantially relax
distributional assumptions for independent outcomes, many applications
involve clustered or longitudinal data requiring inference to account for dependence 
across observations. Generalized linear mixed models (GLMMs) extend GLMs
by introducing random effects to accommodate clustered, longitudinal, or
hierarchical dependence structures \citep{breslow1993approximate}. By modeling
the outcome distribution through both fixed and random effects, GLMMs
provide a flexible and widely used framework for correlated data and have
become standard tools in biomedical and longitudinal data analysis.
However, GLMMs inherit several structural limitations from the classical GLM
framework. First, the outcome distribution conditional on random effects is typically assumed to
belong to a parametric exponential family, so inference remains vulnerable to 
distributional misspecification. Second, regression models condition on random effects, 
and the population--level quantities that may be of interest and are 
obtained by integrating over the random effects are generally 
not available in closed form, particularly under nonlinear links.
Furthermore, dependence in GLMMs is induced indirectly through
latent Gaussian structures, which may restrict the range of attainable
dependence patterns, particularly in settings involving heterogeneous outcome
types or tail--sensitive joint risk assessment
\citep{tuerlinckx2006statistical}.

Marginalized longitudinal models provide an important alternative to conventional
GLMMs by specifying marginal regression effects directly while using
latent or transition structures to represent within--subject dependence.
Early work developed marginalized logistic--normal and multilevel models for
longitudinal outcomes, separating marginal regression parameters from the
dependence model \citep{heagerty1999marginally,heagerty2000marginalized}.
Other work applied this framework to likelihood--based transition models
and richer serial dependence for longitudinal binary or categorical responses
\citep{azzalini1994logistic,heagerty2002marginalized,schildcrout2007marginalized} and to ordinal
longitudinal outcomes \citep{lee2008marginalized}. More recent work has
considered ordinal outcomes with absorbing states
\citep{schildcrout2022modelassisted,di2026modeling}. These models preserve interpretable
marginal effects, but generally rely on more structured outcome distributions
and dependence mechanisms than the continuous--outcome model developed here.

Other methodological developments have addressed different components of
this problem. Robust GLMM formulations have relaxed distributional assumptions
on random effects, and semiparametric or mixture--based extensions have been
proposed to increase flexibility
\citep{kleinman1998semi,ryu2011bayesian,hannah2011dirichlet,antonelli2016mitigating}.
Copula--based marginal regression and mixed--model approaches have made explicit
dependence modeling possible while retaining marginal interpretations, including
Gaussian--copula models for correlated data
\citep{song2009joint,masarotto2012gaussian}. These models, however, generally
use more structured marginal outcome distribution specifications than the DPGLM. \citet{wang2019copula}
develop a closer semiparametric analogue for continuous longitudinal outcomes
by approximating covariate-specific marginal distributions through quantile
regression and modeling intra--subject dependence with a copula. 
\citet{barbanti2024transformation} developed a likelihood-based frequentist framework 
in which flexible monotone transformation functions characterize the marginal outcome 
distributions, while a latent multivariate normal structure captures within-subject dependence.
In a different direction, \citet{smith2010modeling} focus primarily on representing
rich serial dependence through a pair--copula decomposition with separately
specified margins, rather than on GLM--style marginal regression.

These contributions show that flexible marginal modeling and explicit
dependence modeling are both feasible, but they pursue different combinations
of marginal specification, dependence structure, and inferential targets. Our
contribution is to combine the Bayesian nonparametric marginal distribution of
the DPGLM with an explicit copula for repeated measurements. This construction
retains the semiparametric GLM interpretation of the regression effects while
supporting posterior inference for the full marginal distribution and derived
functionals, such as lower-tail probabilities, in longitudinal settings.
Throughout, unless otherwise stated, “marginal” refers to the distribution or mean 
of an individual repeated measurement and the associated regression effects.
The remainder of the paper is organized as follows. Section~\ref{sec:motivating_app}
introduces the motivating Lung Health Study application. Section~\ref{sec:model}
develops the proposed model and posterior specification. Section~\ref{sec:simulation}
reports simulation results, Section~\ref{sec:results} presents the case study
analysis, and Section~\ref{sec:discussion} concludes with a discussion.

\section{Motivating Application} \label{sec:motivating_app}

The proposed semiparametric GLM approach is motivated by analyses of the
National Heart, Lung, and Blood Institute (NHLBI) Lung Health Study (LHS)
\citep{hansel2009leptin}. We consider a binary genetic biomarker (SNP rs177852 in LHS), 
that is encoded as \(-1\) and \(1\). Our main endpoint is FEV1/FVC, the ratio 
of Forced Expiratory Volume in the First second (FEV1) to Forced Vital 
Capacity (FVC), measured after bronchodilator administration and used 
as an indicator of lung function. Patients were followed for five years, 
with up to five approximately annual visits during follow--up. To improve 
interpretability of the time--dependent effects, visit times are centered 
and rescaled to -0.5, -0.25, 0, 0.25, and 0.5. 

For the present analysis, we defined an analytic cohort with available
genotype information and complete measurements of the outcome and smoking--exposure 
covariates. We restricted the data to continuous smokers and excluded observations 
from visits occurring after a patient reported smoking cessation. 
We also required at least two
observed visits per patient so that within--patient dependence could be
estimated. These criteria yielded 2,562 patients and 11,196 observations.
Accordingly, the following summaries and analyses describe continuous smokers
with the required genetic, outcome, and covariate information.

Figure~\ref{fig:lung_health_eda} summarizes the longitudinal outcome profiles.
The thin gray lines represent individual FEV1/FVC trajectories across visits,
revealing substantial between--patient heterogeneity alongside noticeable
within--patient variation over time. The superimposed summary captures the
population--level trend: the median trajectory (solid line with points)
together with interquartile ranges (error bars) at each visit.
The box plot in panel (B) compares the empirical distribution of
FEV1/FVC at each annual visit across biomarker groups, and summarizes how the 
marginal outcome distribution at each visit changes over time.
A post-bronchodilator FEV1/FVC ratio below \(0.7\) 
is a spirometric criterion for persistent airflow obstruction and supports the 
diagnosis of Chronic Obstructive Pulmonary Disease (COPD) in the appropriate 
clinical context \citep{gold_spirometry}. Consequently, the probability of 
FEV1/FVC falling below \(0.7\) is of direct clinical relevance: it quantifies 
the risk that a patient with a given covariate profile exhibits airflow obstruction 
at a particular visit. Estimating this probability therefore provides a clinically 
interpretable summary of the outcome distribution beyond the marginal mean and 
facilitates comparisons of risk across patient characteristics and over time.

The longitudinal measurements exhibited strong within--patient dependence. 
To characterize serial dependence after accounting for the marginal mean 
structure, we fitted a working--independence beta regression model with visit, 
biomarker, and their interaction and computed Pearson residuals. The empirical 
correlations between residuals separated by one through four annual visits were 
0.928, 0.913, 0.894, and 0.873, respectively, indicating substantial residual 
dependence throughout follow--up as expected. Residual variability also increased 
over time, with the empirical residual variance increasing from 0.843 at the 
first visit to 1.28 at the last visit. Together, these features motivate a 
flexible longitudinal model that accommodates both within--patient dependence
and changes in the outcome distribution over time.


\begin{figure}[!htp]
\centering
\begin{subfigure}{0.49\textwidth}
\centering
\includegraphics[width=\linewidth]{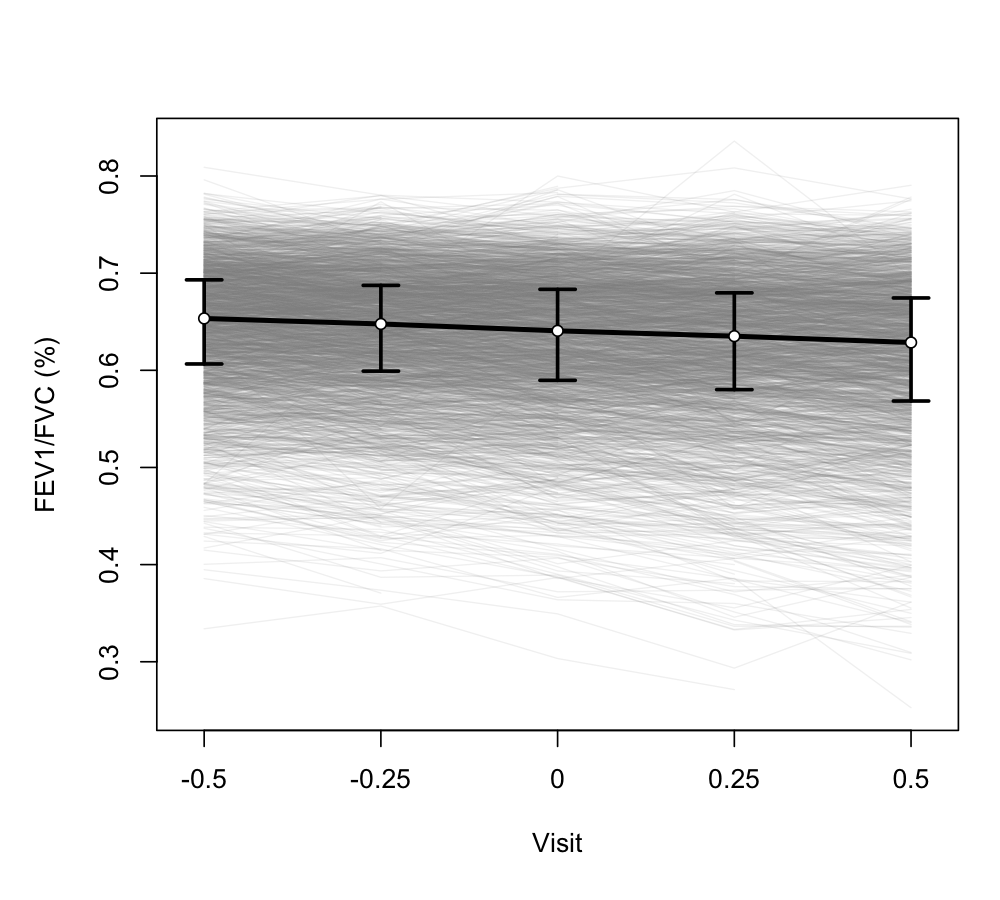}
\caption{Individual FEV1/FVC trajectories by visit.}
\end{subfigure}
\hfill
\begin{subfigure}{0.49\textwidth}
\centering
\includegraphics[width=\linewidth]{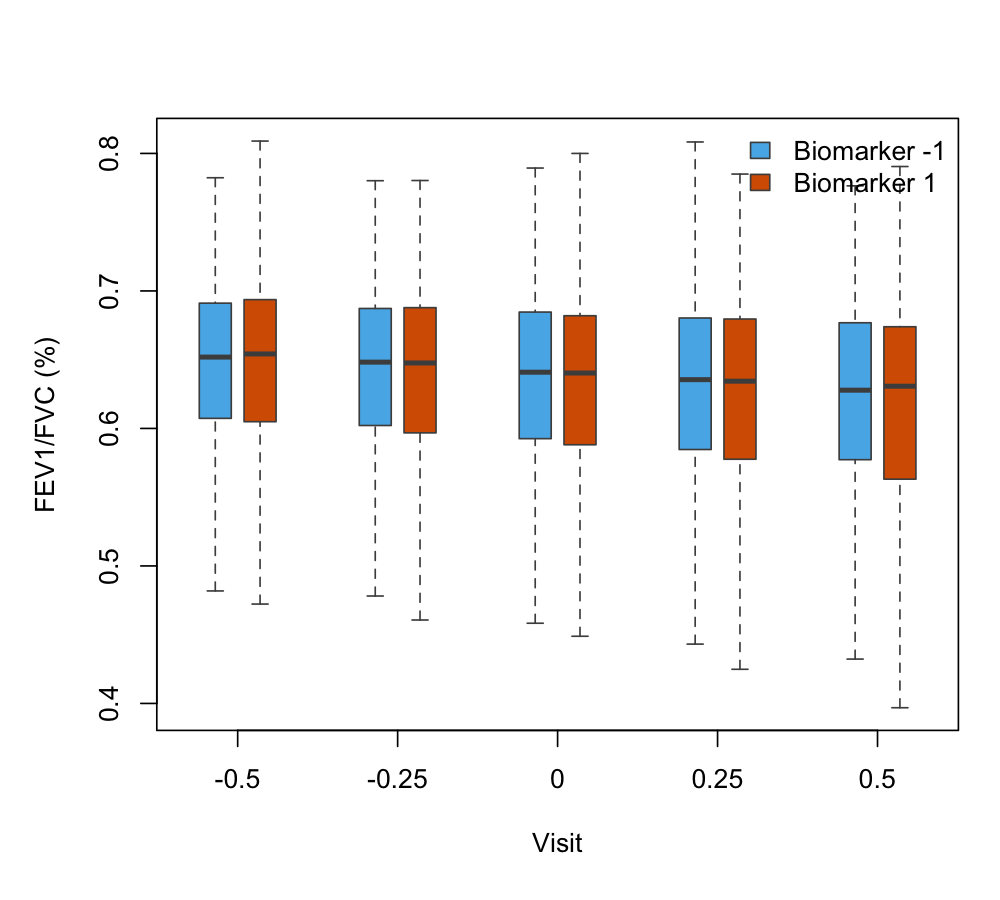}
\caption{Distribution of FEV1/FVC at each visit.}
\end{subfigure}
\caption{Lung Health Study exploratory summaries. Panel (A) shows individual
FEV1/FVC trajectories over visits, with the pointwise median and interquartile
range superimposed. Panel (B) summarizes the empirical FEV1/FVC distributions
by visit and biomarker group.}
\label{fig:lung_health_eda}
\end{figure}

\section{Model} \label{sec:model}

\subsection{DPGLM Marginal Model}

Let $y_{ij}$ denote the continuous outcome for unit $i=1,\ldots,n$ at visit
$j=1,\ldots,n_i$, with covariate vector $\bm{x}_{ij}\in\mathbbm{R}^p$.
Let $\bm{y}_i=(y_{i1},\ldots,y_{in_i})$, $N=\sum_{i=1}^n n_i$, and let
$\YY \subseteq \mathbb{R}$ denote the support of the outcome.
The goal is to model (i) flexible marginal outcome distributions with GLM--type
mean regression and (ii) within--unit dependence across repeated outcomes.
We first briefly review the (marginal) DPGLM for a generic fixed covariate
vector $\bx$. For observation $y = y_{ij}$, substitute $\bx = \bx_{ij}$.
Let $\lambda$ denote an unknown measure on $\YY$, referred to as the baseline measure. 
The DPGLM assumes $y_{ij}$ follows the model $p_{\bx}$ with
\begin{align} \label{eq:sampling_model}
p_{\bx}(y \mid \lambda,\bbeta)
=
\exp\{\theta_{\bx} y - b(\theta_{\bx};\lambda)\}\lambda(dy),
\end{align}
where
\begin{align}
b(\theta;\lambda)
=
\log\int_{\YY}\exp(\theta y)\lambda(dy)
\end{align}
is the log--normalizing constant and $\theta_{\bx}$ is a derived parameter 
\citep{wurm2018semiparametric}. As in an ordinary 
exponential--family GLM, $\theta_{\bx}$ is determined implicitly through the mean 
constraint
\begin{align}
b'(\theta_{\bx};\lambda)
=
\frac{\int_{\YY} y \exp(\theta y) \lambda(dy)}{\int_{\YY}  \exp(\theta y) \lambda(dy)}
=
\mathbb{E}(y\mid \lambda,\bbeta)
=
g^{-1}(\bx^{\top}\bbeta),
\end{align}
where $b'(\theta;\lambda)=\partial b(\theta;\lambda)/\partial \theta$
denotes the derivative of $b$ with respect to $\theta$, and $g$ is a 
user--specified link function making $\theta_{\bx} = \theta_{\bx}(\bbeta, \lambda)$ 
a function of $\bx$, $\bbeta$ and $\lambda$. 
As noted by \citet{rathouz2009generalized}, the model is not identified without
an additional constraint. For any $a\in\mathbb{R}$, define
$\widetilde{\theta}_{\bx} = \theta_{\bx}-a$ and 
$\widetilde{\lambda}(dy) \propto \exp(ay)\lambda(dy).$
Then $(\widetilde{\theta}_{\bx},\widetilde{\lambda})$ produces the same
$p_{\bx}$ in \eqref{eq:sampling_model}, and therefore the same mean, as
$(\theta_{\bx},\lambda)$. We therefore impose
the centering constraint
$\int_{\YY}y\lambda(dy) / \lambda(\YY)=\mu_0,$
where $ \mu_0=\frac{1}{N}\sum_{i}\sum_{j}y_{ij}.$
Equivalently, $b'(0;\lambda)=\mu_0$, so the constraint fixes the origin of
the exponential tilt. Roughly speaking, under this convention, the normalized measure
$\lambda(dy)/\lambda(\YY)$ corresponds to the outcome distribution at the centered 
covariate profile, namely, visit 0 and the midpoint between the two biomarker groups. 

The DPGLM continues the model construction with a nonparametric Bayesian prior
on $\lambda$, assuming a gamma completely random measure (CRM) \citep{kingman1993poisson},
\begin{align}\label{CRM}
\lambda &\sim \mathrm{Gamma\text{-}CRM}(\nu), \\
\nu(ds,dt) &= s^{-1}e^{-s}ds\,\alpha G_0(dt), \nonumber
\end{align}
where $\nu$ is the L\'evy intensity that characterizes the CRM including the
base measure $G_0$ and concentration parameter $\alpha$. We choose the gamma
CRM because, after normalization, it yields a Dirichlet--process prior
\citep{ferguson1973bayesian}, giving a familiar and computationally convenient
nonparametric prior for the baseline distribution. That is, the
normalized measure $\lambda(\cdot) / \lambda(\YY)$ has a Ferguson-Dirichlet 
process prior, $\text{DP}(\alpha, G_0)$. See \citet{lijoi2026ferguson} for a 
recent review. Finally, because $\lambda$ is almost surely discrete, 
the exponentially tilted distribution in \eqref{eq:sampling_model} is discrete. 
When the outcome $y_{ij}$ is continuous we add a convolution with a
continuous kernel by way of adding a discrete latent counterpart $\tilde y_{ij}$ 
and assign it the exponentially tilted distribution 
$p_{\bx}(\tilde y_{ij}\mid \lambda,\bbeta)$ in \eqref{eq:sampling_model}. We 
then generate the continuous outcome $y_{ij}$ from a kernel centered at $\tilde y_{ij}$:
\begin{align}\label{eq:kernel}
y_{ij}\mid \tilde y_{ij} \sim k(y_{ij}\mid \tilde y_{ij}).
\end{align}
Using a symmetric kernel \eqref{eq:kernel} and the tilted distribution 
\eqref{eq:sampling_model} we smooth the 
discrete random measure while preserving the DPGLM mean 
structure. For the observed design point $\bm{x}_{ij}$, let $f_{ij}$ and
$F_{ij}$ denote the marginal density and cumulative distribution function (CDF) 
of $y_{ij}$, respectively.
\begin{equation} \label{eq:fij}
f_{ij}(y\mid \lambda,\bbeta) =
\int_{\YY} k(y\mid \tilde{y})\,p_{\bx}(d \tilde{y}\mid \lambda,\bbeta)
\qquad \text{and} \qquad 
F_{ij}(y\mid \lambda,\bbeta) =
\int_{-\infty}^{y} f_{ij}(t\mid \lambda,\bbeta)\,dt.
\end{equation}
We complete the marginal specification with an independent prior $p(\bbeta)$ on 
the regression coefficients and the non-parametric Bayesian prior
\eqref{CRM} on $\lambda$.

\subsection{Copula Model for Within--Unit Dependence}

The marginal DPGLM defines flexible semiparametric GLM--style inference. Longitudinal 
data, however, require an additional component to account for within--subject dependence. 
We use a copula construction for this purpose. The resulting separation between 
the DPGLM for the marginal distribution and the copula for the within--subject 
dependence is particularly appealing because it preserves the semiparametric 
marginal interpretation of the regression coefficients while allowing flexible 
dependence modeling.

The copula construction is based on Sklar's theorem \citep{sklar1959fonctions}, 
which states that any multivariate distribution can be written in terms of its 
marginal distributions and a copula. Let
\begin{equation}\label{eq:uij}
u_{ij} = F_{ij}(y_{ij}\mid \lambda,\bbeta)
\tag{6a}
\end{equation}
denote the marginal c.d.f. under \eqref{eq:fij}.
For each unit $i$, we model dependence in
$\bm{u}_i=(u_{i1},\ldots,u_{in_i})$ by a parametric copula $C_{\rho}$ with
density $c_{\rho}$ as follows.
Let $F(\by_i \mid \ldots)$ denote the multivariate c.d.f. for
$\by_i$. By Sklar's theorem any joint distribution with given
marginals $f_{ij}$ -- as in \eqref{eq:fij} -- can be written as
$
F(\bm{y}_i\mid \lambda,\bbeta,\rho) = C_{\rho}(u_{i1},\ldots,u_{in_i}),
$
and therefore
$
f(\bm{y}_i\mid \lambda,\bbeta,\rho) = c_{\rho}(u_{i1},\ldots,u_{in_i})
\prod_{j=1}^{n_i} f_{ij}(y_{ij}\mid \lambda,\bbeta).
$
Here $\rho$ is the copula dependence parameter (or vector of
parameters) that indexes the copula function.
This construction leaves the DPGLM margins unchanged and adds a transparent
parameterization of within--subject dependence indexed by $\rho$. The model is 
completed with a prior $p(\rho)$.

Assuming conditional independence across units given
$(\lambda,\bbeta,\rho)$, the joint likelihood is
\begin{align}
L(\lambda,\bbeta,\rho;\bm{y})
=
\prod_{i=1}^n
\left[
c_{\rho}(u_{i1},\ldots,u_{in_i})
\prod_{j=1}^{n_i} f_{ij}(y_{ij}\mid \lambda,\bbeta)
\right],
\end{align}
implying the posterior model
\begin{align}\label{post}
p(\lambda,\bbeta,\rho\mid \bm{y})
\propto
L(\lambda,\bbeta,\rho;\bm{y})\,p(\lambda)p(\bbeta)p(\rho).
\end{align}
For posterior computation, we use the augmented joint model involving the 
latent variables $\tilde{\bm{y}}=(\tilde y_{ij})$:
\begin{align}\label{post-aug}
p(\bm{y},\tilde{\bm{y}},\lambda,\bbeta,\rho)
\propto
\prod_{i=1}^n \left[ c_{\rho}(u_{i1},\ldots,u_{in_i}) 
\prod_{j=1}^{n_i}
k(y_{ij}\mid \tilde y_{ij})
p_{\bm{x}_{ij}}(d\tilde y_{ij}\mid \lambda,\bbeta)\right]
\times
p(\lambda)p(\bbeta)p(\rho).
\end{align}
Here $p(\lambda)$ denotes the gamma--CRM prior on the baseline measure,
$p(\bbeta)$ denotes the regression prior, and $p(\rho)$ is the prior
on the copula dependence parameter. Since the CRM is an infinite--dimensional 
random measure, the CRM prior on $\lambda$ does not admit a finite--dimensional 
density representation, and neither does the corresponding posterior distribution. 
For posterior computation
and to interpret \eqref{post} and \eqref{post-aug},
we approximate the CRM by a finite truncation 
\citep{ferguson1972representation}, yielding a finite--dimensional approximation to 
the original model.

We implement posterior inference using Markov chain Monte Carlo (MCMC)
posterior simulation. To update $\lambda$ we use a Metropolis-Hastings transition
probability. The proposal distribution involves an intractable
normalizing constant, so we use an augmented Metropolis--Hastings update
with an auxiliary random measure. In the resulting augmented state space,
the intractable normalizing constants cancel from the acceptance ratio.
The proposal is sampled using auxiliary variables that decouple the CRM
jumps \citep{james2009posterior,catalano2024unified}. Full MCMC implementation
details are provided in Section~\ref{sec:supp-mcmc} of the Supplementary
Appendix.

\section{Simulation Study}
\label{sec:simulation}

We conducted a simulation study comparing the proposed CDPGLM with a Bayesian
beta mixed-effects model with a subject-specific random intercept (hereafter,
beta mixed model) and two generalized estimating equation (GEE) approaches
using exchangeable and independent working correlation structures. Because the
CDPGLM and GEE models target the marginal mean, the GEE fits provide the
primary benchmarks for regression-coefficient recovery. The beta mixed model's
fixed effects are subject-specific rather than marginal and, under the
nonlinear link, are not directly comparable to the marginal regression
coefficients; it therefore serves primarily as a parametric benchmark for
distributional and lower-tail inference through integration over the random
effect. The simulation was designed to evaluate 
the benefit of modeling the baseline distribution nonparametrically. Because 
the CDPGLM does not restrict the response to a fixed parametric family, it can 
adapt to heterogeneous populations with multimodal or otherwise irregular 
outcome distributions. Such flexibility is especially relevant when estimating 
lower-tail probabilities, which can be sensitive to distributional misspecification. 
We evaluated regression-coefficient bias and empirical standard deviation across 
replicates, together with average estimated uncertainty. Here, average estimated
uncertainty denotes the average posterior standard deviation for the Bayesian
methods and the average robust standard error for the GEE approaches. For the CDPGLM 
and beta mixed model, we also evaluated profile-specific fitted marginal
distributions using the Wasserstein--1 distance and lower-tail probabilities
at a prespecified cutoff. Data were generated using an
Exponential--Tilting (ET) model; further details are provided in Supplementary
Appendix Section~\ref{sec:supp-sim-settings}.

\subsection{Design}

We used a logit link with linear predictor
\begin{align}
\eta_{ij}
=
\beta_0
+
\beta_1 \text{Biomarker}_i
+
\beta_2 \text{Visit}_{ij}
+
\beta_3 \left(\text{Biomarker}_i \times \text{Visit}_{ij}\right),
\end{align}
where $\text{Biomarker}_i \in \{-1,1\}$ with equal probability and five visit times
were equally spaced on $[-0.5,0.5]$. We considered $n \in \{200,400,1000\}$ subjects,
five visits per subject, and 1000 replicated data sets for this configuration. Each 
scheduled observation  was independently missing with probability \(0.20\). We considered a single 
fixed--effect setting, with $\bbeta = (\beta_0, \beta_1, \beta_2, \beta_3)$ given by
\begin{align}
& \bbeta=(0.6, -0.3, -0.3, 0.2). \nonumber
\end{align}

Under the ET data generator, the simulation truth followed the CDPGLM
construction. Under this setting we set the normalized baseline distribution
$\lambda(\cdot)/\lambda(\YY)$ to be the empirical distribution of the LHS data, shown in
Figure~\ref{fig:et_empirical_grid}. Let $\YY$ denote the support of
the empirical distribution in the LHS.
For each observation, we computed the target mean
$\mu_{ij}=\mathrm{logit}^{-1}(\eta_{ij})$ and solved for the canonical
parameter $\theta_{ij}$. This produced a
covariate--specific discrete distribution $F_{ij}$ with weights
proportional to $\exp(\theta_{ij} t)\lambda(t)$. To induce within--subject
dependence, we generated a subject--specific vector of correlated uniforms from
an exchangeable Gaussian copula with $\rho = 0.80$ and mapped those uniforms
through the corresponding inverse CDFs to the observed outcomes
$y_{ij}$ using \eqref{eq:uij}.

For each replicated data set we fit the proposed CDPGLM, a Bayesian beta
regression model with random intercepts, and two GEE models with exchangeable
and independent working correlation structures for regression-coefficient
comparisons. For the CDPGLM, we used independent Gaussian priors $\beta_k \sim N(0,2^2)$ 
for the regression coefficients, a uniform prior over the admissible support 
of the exchangeable Gaussian copula correlation parameter $\rho$,
and total mass $\alpha=1$ and prior mean $G_0=U(0,1)$ for
the gamma CRM on $\lambda$ (i.e., the DP prior on the normalized
baseline $\lambda/\lambda(\YY)$).
For the MCMC transition probability to update the CRM we use (see the
appendix for details of the transition probability) $\delta=2$, $M = 20$
un--linked atoms (that is, atoms $m_\ell$ with $\ty_{ij} \neq m_\ell$ for 
all $i$, $j$, and $\ell = 1, \cdots, M$), 
and we adopted a uniform kernel for \eqref{eq:kernel} with bandwidth chosen by
Silverman's rule, where the interquartile range (IQR) enters through
$c_0 = 0.9 n^{-1/5} \min\{\mathrm{sd}(y), \mathrm{IQR}(y)/1.34\} $
\citep{silverman2018density}.
For the beta mixed model, fit in
\texttt{brms}, we used the same $N(0,2^2)$ prior on the regression
coefficients $\beta_k$, an $\mathrm{Exponential}(1.15)$ prior for the random--intercept
standard deviation, and a $\mathrm{Gamma}(1,0.18)$ prior for the beta
precision parameter. The resulting prior predictive
distributions are shown in Figure~\ref{fig:supp-prior-comparison}.
To fit the CDPGLM we used 1000 MCMC
iterations with 200 burn--in draws. The beta mixed benchmark was fit in
\texttt{brms} with one chain, 2000 iterations, and 500 warmup iterations \citep{burkner2017brms}.

At representative covariate profiles, we compared the fitted marginal
distributions with their true marginal distributions using the Wasserstein--1 distance.
Profile-specific fitted marginal distributions for the beta mixed model were obtained by 
integrating over the random effects using 30-point Gauss–Hermite quadrature.
To assess the downstream impact on lower-tail probabilities, we evaluated
the lower-tail probability at a prespecified threshold $c = 0.7$, 
$\Pr(Y_{\bm{x}}<c \mid \mathcal{D})$  where $\mathcal{D}$ denotes the observed 
data set. We considered six representative profiles, 
$\bm{x} = (\text{Biomarker},\text{Visit}) \in 
\{(-1,-0.5), (1,-0.5), (-1,0), (1,0), (-1,0.5),(1,0.5)\}$.

\subsection{Simulation Results}

Table~\ref{tab:sim_et_bias_sd_summary} reports bias and empirical standard
deviation across simulation replicates, along with average estimated uncertainty,
by coefficient, sample size, and method. In this setting,
bias is generally small and tends to decrease as the sample size increases
for the CDPGLM and GEE approaches. The beta mixed model is an exception,
retaining some coefficient-specific bias as the sample size increases.
Empirical variability and reported uncertainty summaries are broadly comparable
across sample sizes. Because the GEE approaches provide
the primary benchmarks for marginal regression recovery, the CDPGLM and GEE
results are the most directly comparable. Both GEE exchangeable and the CDPGLM
show generally small bias and similar empirical variability once the sample size 
reaches 1000 in this setting.
Panel C of Table~\ref{tab:sim_et_bias_sd_summary} shows that the relative 
empirical variability is close to one for the time-fixed biomarker coefficient, 
whereas the time-varying visit and biomarker-by-visit interaction coefficients 
exhibit lower variability relative to GEE independence. This pattern is 
consistent with the repeated-measures design, which provides more information 
for estimating within-subject time effects.

The models differed more clearly in their recovery of profile-specific
marginal distributions. Figure~\ref{fig:sim_et_marginal_distribution}
compares fitted marginal distributions with the truth using Wasserstein--1
error. The CDPGLM has smaller median errors across all six profiles at
$n=200$, $400$, and $1000$. Figure~\ref{fig:sim_et_tail_prob} shows 
corresponding lower-tail probabilities. At $n=200$, the beta mixed model 
is closer to the truth for the profiles 
$(\text{biomarker},\text{visit})=(1,-0.5)$ and $(1,0)$, and generally 
exhibits lower across-replicate variability. This small-sample advantage 
does not persist as the sample size increases. As $n$ increases, the CDPGLM 
estimates move toward the truth and become less variable, whereas the beta 
mixed model retains systematic discrepancies from the truth. By \(n=1000\), 
the CDPGLM is closer to the truth across the displayed profiles, with 
substantially reduced across-replicate variability. Taken together,
Table~\ref{tab:sim_et_bias_sd_summary} and
Figures~\ref{fig:sim_et_marginal_distribution} and
\ref{fig:sim_et_tail_prob} show broadly comparable regression recovery among
the methods targeting marginal effects, but a clearer large-sample advantage
for the CDPGLM over the beta mixed model in estimating profile-specific
marginal distributions and their lower-tail probabilities.

Supplementary Figure~\ref{fig:sim_et_within_patient_correlations} shows that 
the CDPGLM and exchangeable GEE generally recover the true within-patient 
response-scale correlations across biomarker profiles and visit pairs. The 
estimates become more concentrated around the truth as the sample size increases.
Supplementary Figure~\ref{fig:sim_et_regression_rmse} reports regression
root mean square error (RMSE) for the CDPGLM, beta mixed model, and both GEE
approaches. The beta mixed model results are included for completeness, but
its subject-specific regression coefficients target a different estimand from
the marginal coefficients of the CDPGLM and GEE models and should therefore
not be interpreted as a direct regression benchmark.

\begin{table}[!htp]
\centering
\caption{ET simulation results by sample size, coefficient, and method. Panel A
reports percent bias; Panel B reports the ratio of average estimated uncertainty
to the across-replicate standard deviation, where estimated uncertainty is the
posterior standard deviation for Bayesian methods and the robust standard error
for GEE methods; and Panel C reports the across-
replicate standard deviation relative to GEE independence. Values are rounded
to two decimal places.}
\label{tab:sim_et_bias_sd_summary}
\scriptsize
\setlength{\tabcolsep}{3pt}
\parbox{\textwidth}{\centering\normalsize\textbf{Panel A: Percent bias}}\par
\resizebox{\textwidth}{!}{%
\begin{tabular}{ll*{4}{>{\raggedleft\arraybackslash}p{2.6cm}}}
\toprule
Method & $n$ & Intercept & Biomarker & Visit & \makebox[2.6cm][c]{Biomarker $\times$ visit} \\
\midrule
\multirow{3}{*}{CDPGLM} & 200 & 1.92 & -1.90 & -1.06 & -0.79 \\
 & 400 & 0.95 & -1.52 & -0.35 & -0.48 \\
 & 1000 & 0.29 & -0.70 & -0.10 & -0.05 \\
\addlinespace
\multirow{3}{*}{Beta mixed} & 200 & 1.19 & -0.05 & 0.83 & 0.12 \\
 & 400 & 1.20 & 0.12 & 1.02 & -0.21 \\
 & 1000 & 1.15 & 0.10 & 0.93 & -0.12 \\
\addlinespace
\multirow{3}{*}{\shortstack{GEE\\exchangeable}} & 200 & 0.05 & -0.04 & 0.03 & 0.27 \\
 & 400 & 0.02 & 0.09 & 0.13 & -0.12 \\
 & 1000 & -0.05 & 0.06 & 0.02 & -0.05 \\
\addlinespace
\multirow{3}{*}{\shortstack{GEE\\independent}} & 200 & 0.01 & 0.04 & 0.09 & 0.07 \\
 & 400 & 0.00 & 0.14 & -0.04 & 0.03 \\
 & 1000 & -0.06 & 0.10 & 0.14 & -0.11 \\
\bottomrule
\end{tabular}%
}
\parbox{\textwidth}{\centering\normalsize\textbf{Panel B: Average estimated uncertainty divided by across-replicate SD}}\par
\resizebox{\textwidth}{!}{%
\begin{tabular}{ll*{4}{>{\raggedleft\arraybackslash}p{2.6cm}}}
\toprule
Method & $n$ & Intercept & Biomarker & Visit & \makebox[2.6cm][c]{Biomarker $\times$ visit} \\
\midrule
\multirow{3}{*}{CDPGLM} & 200 & 1.00 & 1.01 & 1.09 & 1.09 \\
 & 400 & 0.96 & 1.02 & 1.08 & 1.11 \\
 & 1000 & 0.97 & 1.01 & 0.96 & 1.05 \\
\addlinespace
\multirow{3}{*}{Beta mixed} & 200 & 1.00 & 1.00 & 1.04 & 1.03 \\
 & 400 & 1.00 & 0.98 & 1.01 & 1.04 \\
 & 1000 & 0.99 & 1.00 & 0.93 & 0.98 \\
\addlinespace
\multirow{3}{*}{\shortstack{GEE\\exchangeable}} & 200 & 0.99 & 1.00 & 1.00 & 0.99 \\
 & 400 & 1.01 & 0.98 & 1.00 & 1.03 \\
 & 1000 & 0.98 & 1.00 & 0.95 & 0.99 \\
\addlinespace
\multirow{3}{*}{\shortstack{GEE\\independent}} & 200 & 0.99 & 0.99 & 0.97 & 0.97 \\
 & 400 & 1.00 & 0.99 & 1.01 & 1.04 \\
 & 1000 & 0.98 & 1.01 & 0.96 & 0.95 \\
\bottomrule
\end{tabular}%
}
\parbox{\textwidth}{\centering\normalsize\textbf{Panel C: Across-replicate SD relative to GEE independence}}\par
\resizebox{\textwidth}{!}{%
\begin{tabular}{ll*{4}{>{\raggedleft\arraybackslash}p{2.6cm}}}
\toprule
Method & $n$ & Intercept & Biomarker & Visit & \makebox[2.6cm][c]{Biomarker $\times$ visit} \\
\midrule
\multirow{3}{*}{CDPGLM} & 200 & 1.11 & 0.99 & 0.83 & 0.82 \\
 & 400 & 1.08 & 0.94 & 0.78 & 0.79 \\
 & 1000 & 1.05 & 0.96 & 0.79 & 0.73 \\
\addlinespace
\multirow{3}{*}{Beta mixed} & 200 & 1.00 & 0.99 & 0.78 & 0.78 \\
 & 400 & 0.99 & 1.00 & 0.81 & 0.82 \\
 & 1000 & 0.99 & 1.01 & 0.81 & 0.77 \\
\addlinespace
\multirow{3}{*}{\shortstack{GEE\\exchangeable}} & 200 & 0.98 & 0.98 & 0.77 & 0.77 \\
 & 400 & 0.97 & 0.98 & 0.80 & 0.81 \\
 & 1000 & 0.98 & 0.99 & 0.80 & 0.76 \\
\bottomrule
\end{tabular}%
}
\end{table}

\begin{figure}[!htp]
\centering
\includegraphics[width=\textwidth]{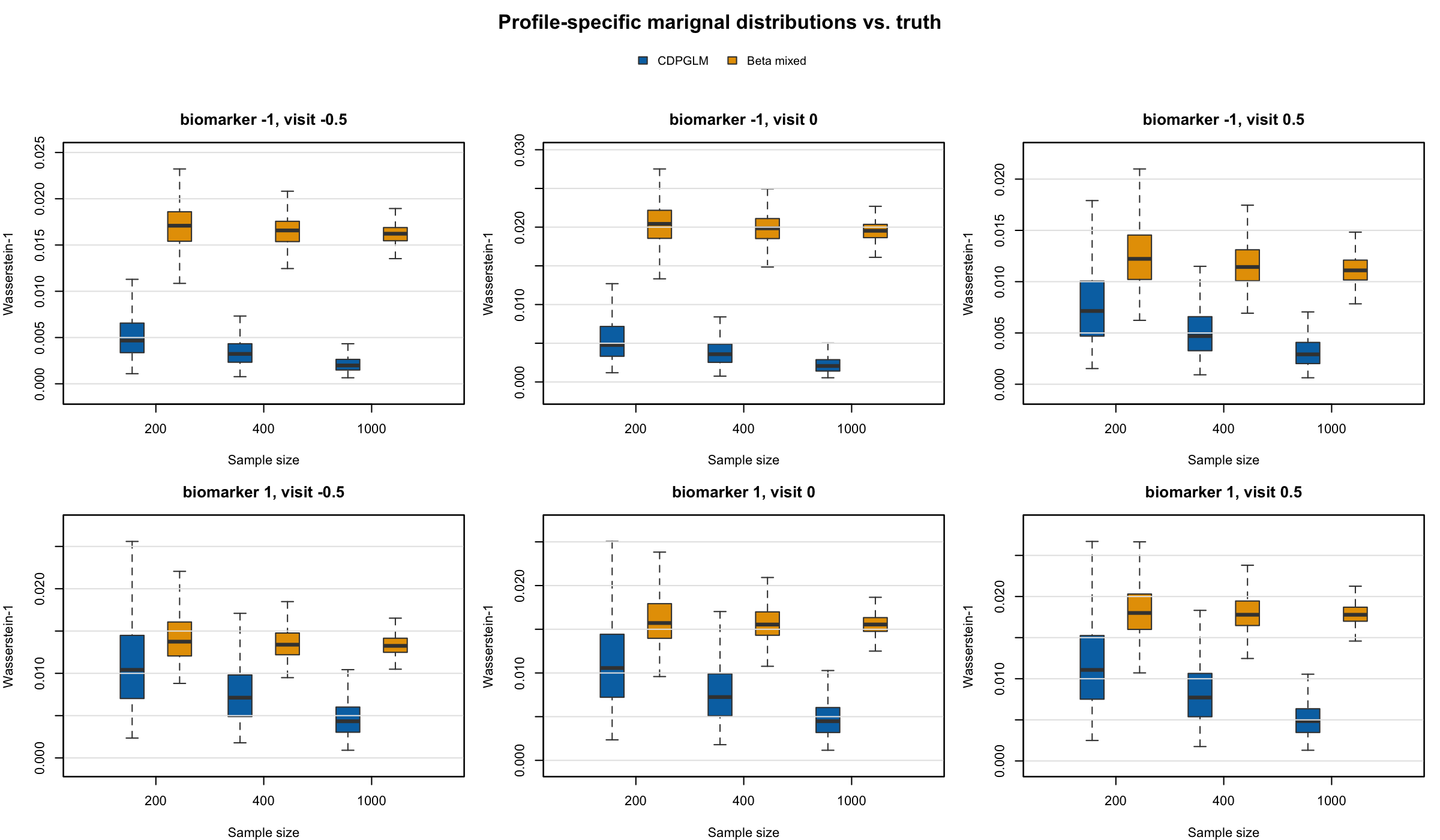}
\caption{Wasserstein--1 errors for fitted marginal
distributions at representative profiles, by biomarker (row), visit (column), 
and sample size (x-axis).}
\label{fig:sim_et_marginal_distribution}
\end{figure}

\begin{figure}[!htp]
\centering
\includegraphics[width=\textwidth]{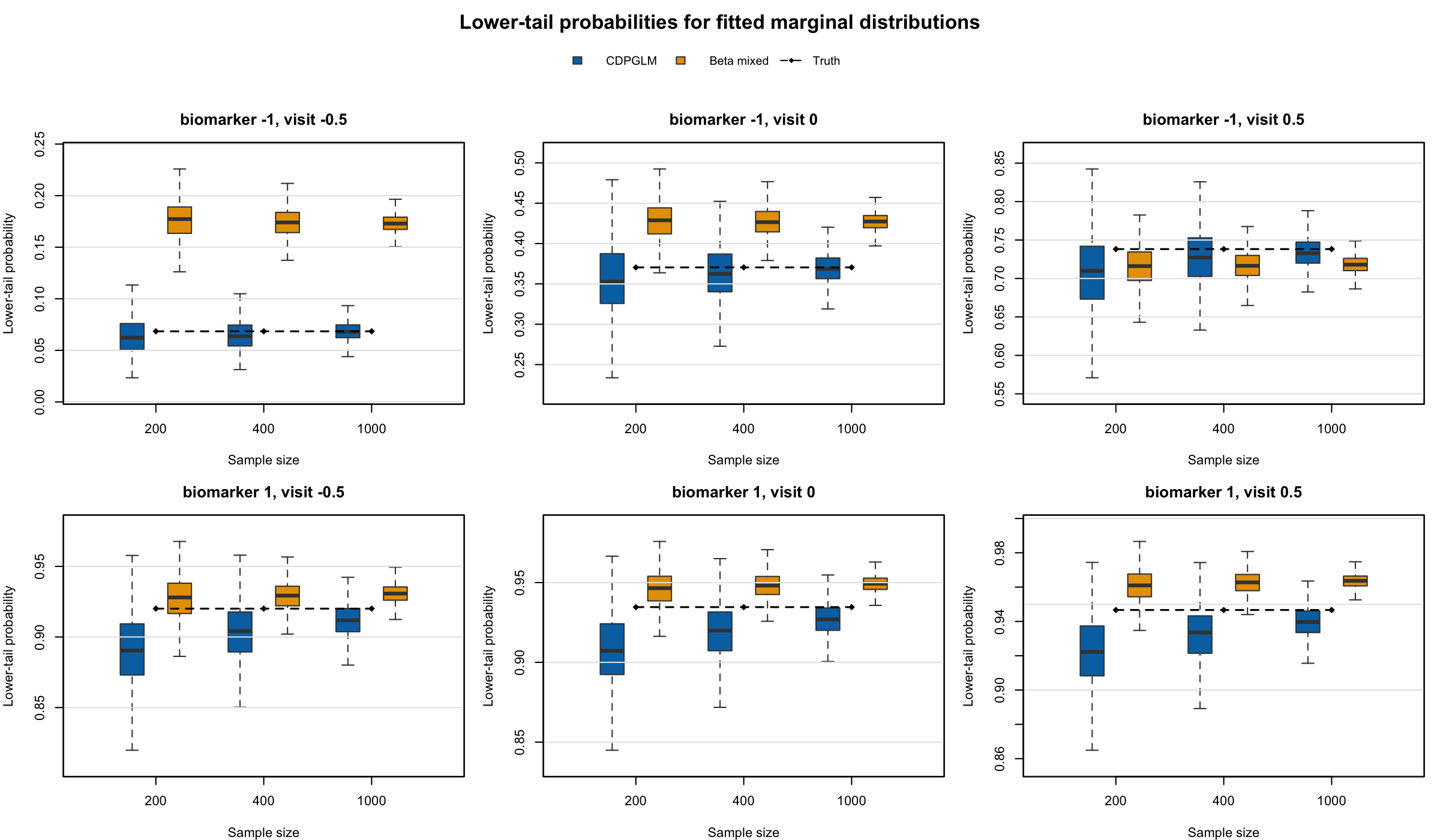}
\caption{Estimated lower-tail probabilities from the fitted marginal
distributions, $\Pr(Y_{\bm{x}}<0.7\mid\mathcal{D})$, at representative profiles, by biomarker (row), visit (column), 
and sample size (x-axis); dashed lines mark the truth.}
\label{fig:sim_et_tail_prob}
\end{figure}

\section{Results -- Lung Health Study}
\label{sec:results}

We fitted the CDPGLM to the Lung Health Study data introduced in
Section~\ref{sec:motivating_app} and compared it with a beta mixed--effects
model with a subject--specific random intercept. Both models
used the same fixed--effect structure, including biomarker, visit, age, sex,
BMI, pack-years, study-site indicators, and the biomarker--by--visit
interaction. The beta mixed model uses a parametric beta response distribution
and accounts for within--subject dependence through a random intercept. The
CDPGLM combines a nonparametric marginal model with a copula dependence
structure. Additional case-study and
covariate details are provided in Section~\ref{sec:supp-case-study} and
Tables~\ref{tab:supp-case-study-covariates}--\ref{tab:supp-case-study-site} of
the supplementary materials.

Table~\ref{tab:case_study_summary} reports selected posterior summaries. With
the logit link, the CDPGLM coefficients describe changes in the logit of
the marginal expected outcome associated with one-unit covariate changes,
holding the other covariates fixed. The beta mixed-model coefficients instead
describe subject-specific changes conditional on the random intercept and are
therefore not directly comparable in magnitude to the CDPGLM coefficients.
Because visit was scaled from $-0.5$ to $0.5$, a one-unit increase spans the
observed follow-up period. Both models estimated a negative visit effect and a
small negative biomarker--by--visit interaction, indicating qualitative
agreement in the direction of the temporal trend and its modification by
biomarker. Because the coefficient magnitudes target different estimands, they
should not be interpreted as estimates of the same effect. Full coefficient
summaries are reported in Supplementary
Table~\ref{tab:supp-case-study-coefficients}.

\begin{table}[!htp]
\centering
\begin{threeparttable}
\caption{Lung Health Study: selected posterior summaries for regression
parameters and lower-tail probabilities under two alternative models.}
\label{tab:case_study_summary}
\footnotesize
\setlength{\tabcolsep}{3pt}
\begin{tabular}{lcccccc}
\hline
Method & Intercept & Biomarker & Visit & Biomarker $\times$ visit
& $D_{-1}$ & $D_{1}$ \\
\hline
Beta mixed &
0.463 (0.018) & -0.007 (0.005) & -0.129 (0.003) & -0.009 (0.003) & 0.926 (0.009) & 0.936 (0.008) \\
CDPGLM &
0.476 (0.017) & -0.005 (0.005) & -0.122 (0.003) & -0.010 (0.003) & 0.905 (0.008) & 0.913 (0.007) \\
\hline
\end{tabular}
\begin{tablenotes}
\item Entries are posterior mean (posterior standard deviation). Here
$D_{-1}=\Pr(Y < 0.7 \mid \mathrm{biomarker}=-1,\mathrm{visit}=0.5,\mathcal{D})$ and
$D_{1}=\Pr(Y < 0.7 \mid
\mathrm{biomarker}=1,\mathrm{visit}=0.5,\mathcal{D})$ are lower-tail
probabilities.
\end{tablenotes}
\end{threeparttable}
\end{table}

We next assessed the fitted marginal distributions. For each CDPGLM posterior
draw, we evaluated the covariate-specific tilted density and then averaged the
density curves across post--burn--in draws. For the beta
mixed model, we integrated numerically over the random-intercept distribution
to obtain the marginal density. Figure~\ref{fig:case_study_density_comparison}
compares these fitted densities with the empirical densities at visit $0.5$.
For both biomarker values, the beta mixed density is shifted to the left of the
empirical primary mode, whereas the CDPGLM better reproduces the empirical
shape and aligns more closely with the location of the primary mode.

The last two columns of Table~\ref{tab:case_study_summary} report lower-tail
probabilities at the clinically relevant threshold $0.7$. For biomarker values
$-1$ and $1$, the beta mixed model estimated probabilities of $0.926$ and
$0.936$, whereas the CDPGLM estimated $0.905$ and $0.913$, respectively. The
CDPGLM estimates were lower for both biomarker profiles. Overall, the two
models produced broadly similar
mean-regression results but discernible differences in the fitted marginal
distributions and associated lower-tail probabilities. The CDPGLM combines a
flexible marginal distribution with explicit modeling of within--subject
dependence.

\begin{figure}[!htp]
\centering
\begin{subfigure}{0.49\textwidth}
\centering
\includegraphics[width=\linewidth]{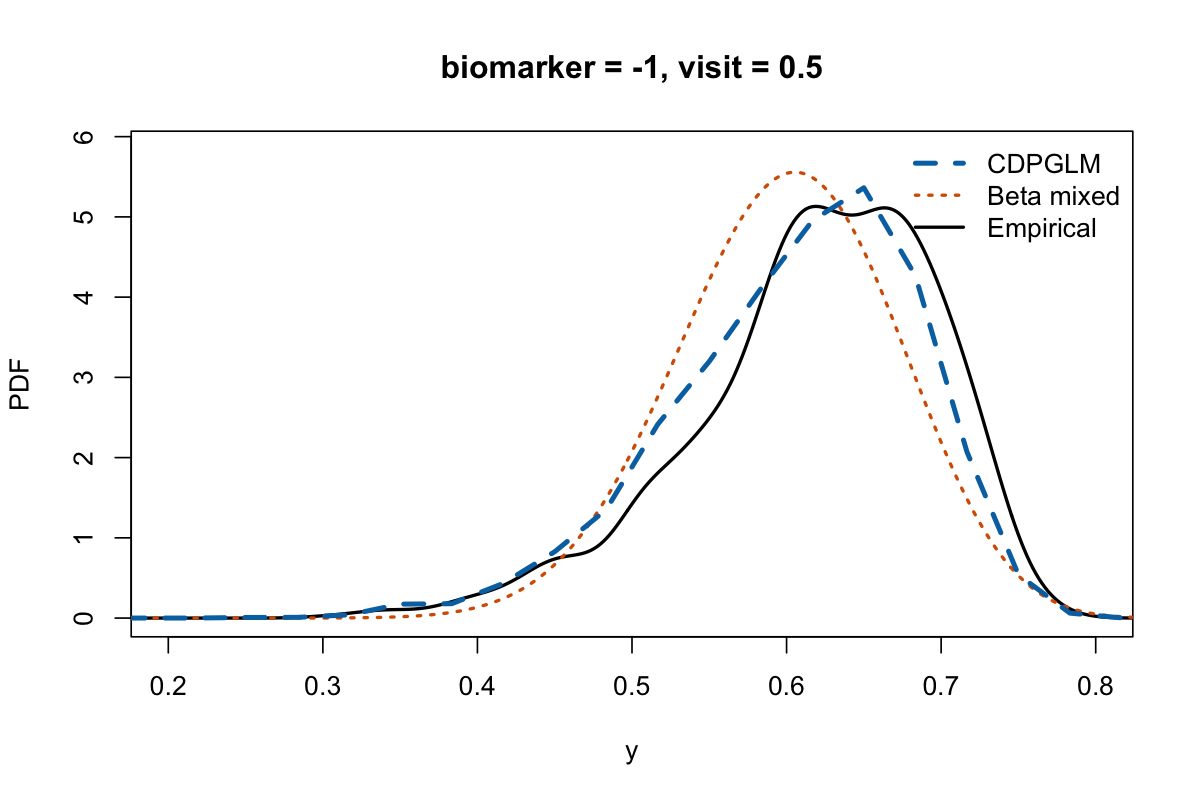}
\caption{Biomarker $-1$ at visit 0.5.}
\end{subfigure}
\hfill
\begin{subfigure}{0.49\textwidth}
\centering
\includegraphics[width=\linewidth]{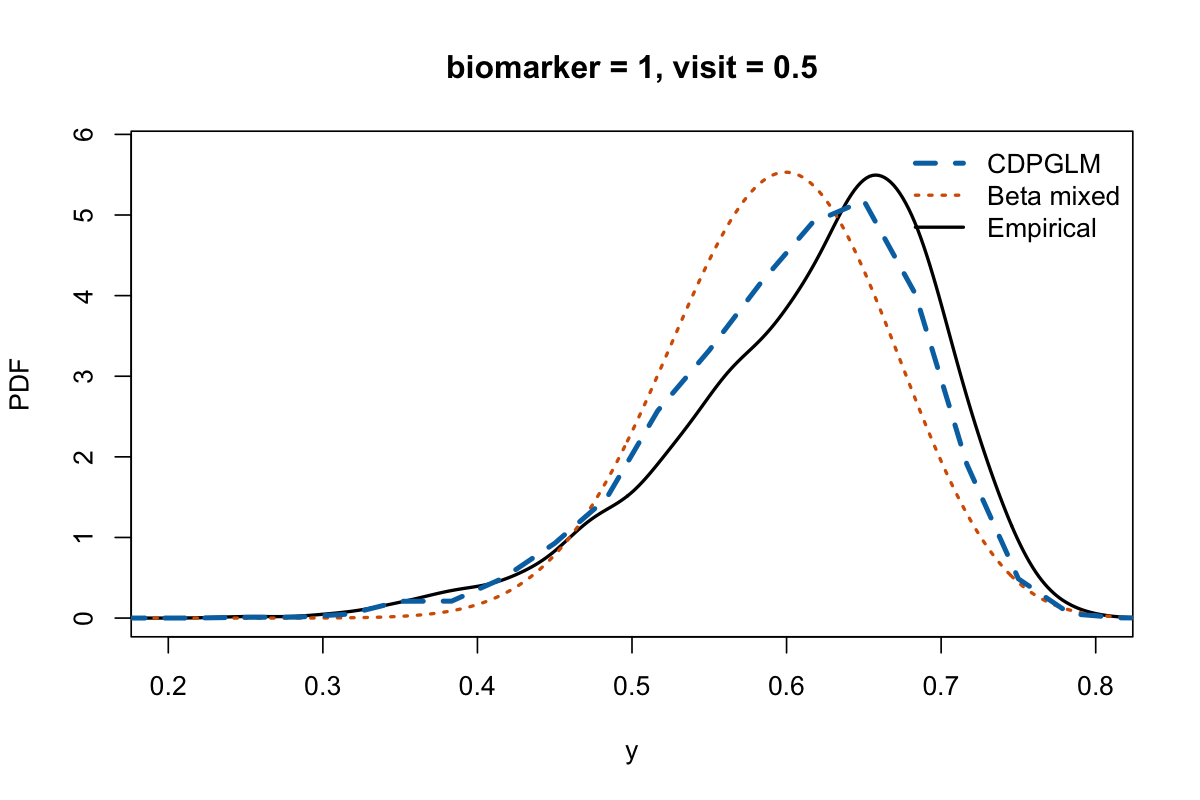}
\caption{Biomarker $1$ at visit 0.5.}
\end{subfigure}
\caption{Lung Health Study fitted density comparisons. Panels (A) and (B)
compare the fitted densities under the beta mixed model and CDPGLM with
the empirical probability density function at visit 0.5 for biomarker values
$-1$ and $1$, respectively.}
\label{fig:case_study_density_comparison}
\end{figure}

\section{Discussion}
\label{sec:discussion}

We proposed CDPGLM as a copula--based extension of DPGLM for longitudinal
continuous outcomes. The central idea is to retain the semiparametric marginal
GLM structure of the DPGLM while using a copula to model within--subject
dependence. This separation is important: the copula captures serial
association, while the DPGLM component continues to
define the marginal density, regression mean, and any summary of interest such as 
lower-tail probabilities.
That is, dependence can be added without changing the marginal
interpretation of the regression model.

The simulation results clarify where this flexibility matters. Under the ET
generator, the CDPGLM and GEE approaches showed broadly similar recovery of
marginal regression effects, whereas
the beta mixed model targets subject-specific rather than marginal
coefficients and is therefore not a direct regression benchmark. The fitted
marginal distributions nevertheless differed substantially between the CDPGLM
and beta mixed model. This gap
between marginal regression and distributional inference under a flexible
baseline, on the one hand, and the subject-specific, parametric beta mixed
model, on the other, motivates the CDPGLM. The finding is also consistent with
the asymptotic orthogonality between the mean--model parameters and the
baseline distribution in semiparametric GLMs \citep{huang2017orthogonality}:
regression estimation can remain accurate even when distributional recovery
differs between models. The benefit of the nonparametric baseline lies in
inference for functionals of the baseline distribution, particularly
lower-tail probabilities.

Several limitations and opportunities for future work remain. First, the 
dependence structure is prespecified. In this paper, we used an exchangeable 
Gaussian copula as a parsimonious representation of within--subject dependence. 
Although the marginal model is specified separately from the dependence 
structure, misspecification of the latter may still affect uncertainty 
quantification. An important extension is therefore to consider more flexible 
or data-adaptive dependence models, including AR(1) and unstructured correlation 
structures as well as alternative copula families, and to systematically investigate 
the robustness of posterior uncertainty to dependence misspecification.

Second, the present study focused on bounded continuous outcomes and a particular 
specification of the baseline distribution model. Future work could extend the 
CDPGLM framework to other outcome types, including discrete and mixed outcomes, 
and investigate sensitivity to modeling choices such as the smoothing kernel 
and baseline prior. Developing principled approaches for selecting or averaging 
over alternative dependence models would also be useful.

Third, the current framework assumes that subjects are independently sampled 
from the target population. Extending CDPGLM to more complex sampling designs is 
another important direction. One example is outcome-dependent sampling, in which 
inclusion probabilities depend on the observed outcome. In such settings, the 
sampling mechanism would need to be incorporated into the likelihood or otherwise 
accounted for to obtain valid population-level inference for the marginal 
distribution and its derived functionals.

\newpage\clearpage
\bibliographystyle{apalike}
\bibliography{cdpglm}

\section*{Acknowledgments}
With sadness we note that one of the authors, Paul Rathouz, passed away 
in December 2025. Paul was a wonderful and highly collegial member of our 
scientific community and is dearly missed. We dedicate this article to his memory.

We have no conflicts of interest to
declare.
We used artificial intelligence (AI)-based tools to assist with language polishing and editorial
refinement of the manuscript.

\section*{Supplementary Materials}
Web appendices, additional figures, and the R code for the simulation study
and data analysis are available with this paper at the Biometrics website on
Oxford Academic. These materials include the full simulation outputs summarized
in Section~\ref{sec:simulation}, along with scripts for reproducing the
reported tables and figures. The supplementary appendix contains the MCMC and
implementation details in Section~\ref{sec:supp-mcmc}, a prior predictive
comparison in Section~\ref{sec:supp-prior-comparison}, and expanded simulation
settings and additional ET diagnostics in
Section~\ref{sec:supp-sim-settings}.

\section*{Data Availability Statement}
The case--study data used in the application are not publicly available. This
does not affect the validity of the paper, because the methodological
development and simulation results do not depend on public release of the
application data. The baseline--grid mass used in the simulation will be shared
in the Supporting Information.
\clearpage
\begin{center}
{\Large\bfseries Supplementary Appendix}
\end{center}

\renewcommand{\thesection}{S.\arabic{section}}
\renewcommand{\thesubsection}{S.\arabic{section}.\arabic{subsection}}
\renewcommand{\thesubsubsection}{S.\arabic{section}.\arabic{subsection}.\arabic{subsubsection}}
\renewcommand{\theHsection}{supp.\arabic{section}}
\renewcommand{\theHsubsection}{supp.\arabic{section}.\arabic{subsection}}
\renewcommand{\theHsubsubsection}{supp.\arabic{section}.\arabic{subsection}.\arabic{subsubsection}}
\setcounter{section}{0}
\setcounter{table}{0}
\renewcommand{\thetable}{S.\arabic{table}}
\renewcommand{\theHtable}{supp.\arabic{table}}
\setcounter{equation}{0}
\renewcommand{\theequation}{S.\arabic{equation}}
\renewcommand{\theHequation}{supp.\arabic{equation}}
\setcounter{figure}{0}
\renewcommand{\thefigure}{S.\arabic{figure}}
\renewcommand{\theHfigure}{supp.\arabic{figure}}
This supplement collects additional material referenced in the main manuscript,
including Markov chain Monte Carlo (MCMC) and implementation details, a prior
predictive comparison of the CDPGLM and beta mixed model, an expanded
description of the simulation design, and additional simulation results for
the ET generator.

\section{MCMC and Implementation Details}
\label{sec:supp-mcmc}

\subsection{Notation}

The marginal and copula models are defined in Section~\ref{sec:model} of the
main manuscript. For the MCMC description below, we only record the posterior
kernels that are updated by the sampler. Let
\begin{equation}
    \label{eq:supp-bth}
    b(\theta_{ij};\lambda)
    \equiv
    \log \int_{\YY} \exp(\theta_{ij} \ty)\lambda(d\ty),
\end{equation}
and
\begin{equation}
    \label{eq:supp-p_x}
    p_{\bx}(\ty_{ij} \mid \lambda, \bbeta)
    =
    \exp\{\theta_{ij} \ty_{ij} - b(\theta_{ij};\lambda)\}\lambda(d\ty_{ij}),
\end{equation}
where $\theta_{ij}$ is the canonical parameter determined by the mean constraint
in the main text. The observed--data marginal density is
\begin{equation}
    \label{eq:supp-marginal-y}
    f_{ij}(y_{ij} \mid \lambda, \bbeta)
    =
    \int k(y_{ij} \mid \ty_{ij}) p_{\bx}(\ty_{ij} \mid \lambda, \bbeta) d\ty_{ij},
\end{equation}
with corresponding cumulative distribution function (CDF) $F_{ij}$. For each unit, define
$F_j = F_{ij}(y_{ij} \mid \lambda, \bbeta)$ and
\begin{equation}
    \label{eq:supp-copula-density}
    c(\boldsymbol{y}_i, \lambda, \bbeta, \rho)
    \equiv
    \frac{\partial^{n_i} C_{\rho}}{\partial F_1 \cdots \partial F_{n_i}}.
\end{equation}
Here, the notation $c(\boldsymbol{y}_i, \lambda, \bbeta, \rho)$ differs from 
$c_{\rho}(\cdots)$ in the main text to make its dependence on $\lambda$, 
$\bbeta$, and $\rho$ explicit.
The posterior kernel based on the observed outcomes is then
\begin{equation}
    \label{eq:supp-joint-simple}
    p(\by, \lambda, \bbeta, \rho)
    \propto
    \left\{
        \prod_i
        \left[
            c(\boldsymbol{y}_i, \lambda, \bbeta, \rho)
            \prod_j
            \int k(y_{ij} \mid \ty_{ij})
            p_{\bx}(\ty_{ij} \mid \lambda, \bbeta) d\ty_{ij}
        \right]
    \right\}
    p(\lambda)p(\bbeta)p(\rho),
\end{equation}
and augmenting with latent variables $\bty$ gives
\begin{equation}
    \label{eq:supp-joint-latent}
    p(\by, \bty, \lambda, \bbeta, \rho)
    \propto
    \left\{
        \prod_{i = 1}^n
        \left[
            c(\boldsymbol{y}_i, \lambda, \bbeta, \rho)
            \prod_{j = 1}^{n_i}
            k(y_{ij} \mid \ty_{ij})
            p_{\bx}(\ty_{ij} \mid \lambda, \bbeta) 
        \right]
    \right\}
    p(\lambda)p(\bbeta)p(\rho).
\end{equation}
For computation, it is convenient to combine the copula contribution and the
tilted density through
\begin{equation} \label{eq:tilte_factor}
    \exp\left[
        \theta_{ij} \ty_{ij}
        +
        \frac{1}{n_i}
        \log\left\{
            c(\by_i, \lambda, \bbeta, \rho)
        \right\}
    \right]
    \equiv
    \exp\left[
        \theta_{ij} \ty_{ij}
        +
        h(\by_i, \lambda, \bbeta, \rho)
    \right].
\end{equation}

\subsection{MCMC Algorithm}
\label{sec:supp-mcmc-updates}
Details of the MCMC algorithm and associated results
follow the development in \citet{alam2025dpglm}. The main difference in the
present MCMC implementation is the additional term $h$ in the tilting factor
in \eqref{eq:tilte_factor}, arising from the copula component. In addition,
when updating the random measure $\lambda$, we introduce an augmented conditional 
target distribution to improve the acceptance rate. The Markov chain cycles 
through four sets of transition probabilities:
\begin{enumerate}
    \item update the regression coefficients $\bbeta$;
    \item update the latent variables $\bty$;
    \item update the copula parameter $\rho$;
    \item update the random measure $\lambda$.
\end{enumerate}
The first three blocks are standard conditional updates and are summarized in
Section~\ref{sec:supp-standard-mcmc-blocks}. The fourth block,
which updates the CRM, is described separately in
Section~\ref{sec:supp-lambda-update}.

\subsection{Standard MCMC blocks}
\label{sec:supp-standard-mcmc-blocks}

\paragraph{Update $\bbeta$.}
The conditional target posterior distribution for $\bbeta$ is
\begin{equation}
    \begin{aligned}
        \pi(\bbeta \mid \by, \bty, \lambda, \rho)
        \propto &
        \left\{
            \prod_{i = 1}^n
            \left[
                c(\boldsymbol{y}_i, \lambda, \bbeta, \rho)
                \prod_{j = 1}^{n_i}
                p_{\bx}(\ty_{ij} \mid \lambda, \bbeta)
            \right]
        \right\}
        p(\bbeta) \\
        = &
        \left\{
            \prod_{i = 1}^n \prod_{j = 1}^{n_i}
            \exp\left[
                \theta_{ij} \ty_{ij}
                +
                h(\by_i, \lambda, \bbeta, \rho)
                -
                b(\theta_{ij};\lambda)
            \right]
            \lambda(\ty_{ij})
        \right\}
        p(\bbeta),
    \end{aligned}
\end{equation}
where $\theta_{ij}$ depends on $\bbeta$, $\bm{x}_{ij}$, and $\lambda$.

\paragraph{Update $\bty$.}
Using \eqref{eq:supp-joint-latent}, we find the conditional posterior distribution 
for a single $\ty_{ij}$
\begin{equation}
    \pi(\ty_{ij} \mid \by, \lambda, \bbeta, \rho)
    \propto
    k(y_{ij} \mid \ty_{ij}) p_{\bx}(\ty_{ij} \mid \lambda, \bbeta).
\end{equation}

\paragraph{Update $\rho$.}
Using \eqref{eq:supp-joint-simple},
\begin{equation}
    \pi(\rho \mid \by, \lambda, \bbeta)
    \propto
    \prod_i f(\boldsymbol{y}_i \mid \lambda, \bbeta, \rho) p(\rho)
    \propto
    \prod_i c(\boldsymbol{y}_i, \lambda, \bbeta, \rho) p(\rho).
\end{equation}

\subsection{\texorpdfstring{$\lambda$}{lambda} Update}
\label{sec:supp-lambda-update}

The conditional target posterior distribution of $\lambda$ is given by
\begin{equation} \label{eq:supp-poster-crm}
    \pi(\lambda \mid \by, \bty, \bbeta, \rho)
    \propto
    \left\{
        \prod_i \prod_j
        \exp\left[
            \theta_{ij} \ty_{ij}
            +
            h(\by_i, \lambda, \bbeta, \rho)
            -
            b(\theta_{ij};\lambda)
        \right]
        \lambda(\ty_{ij})
    \right\}
    p(\lambda).
\end{equation}
Define
\begin{equation*}
    h_{ij}(\lambda)
    \equiv
    h(\by_i, \lambda, \bbeta, \rho),
    \qquad
    \bxi(\lambda)
    \equiv
    \{\theta_{ij}(\lambda), h_{ij}(\lambda) : i = 1, \ldots, n,\;
    j = 1, \ldots, n_i\}.
\end{equation*}
For any random measure $\lambda$ and \textbf{fixed} $\bxi$, define the
following auxiliary distribution:
\begin{equation} \label{eq:supp-poster-crm-new}
    \Pi(\lambda \mid \by, \bty, \bbeta, \rho, \bxi) =  
    \left\{
        \prod_i \prod_j
        \exp\left[
            \theta_{ij} \ty_{ij}
            +
            h_{ij}
            -
            b(\theta_{ij};\lambda)
        \right]
        \lambda(\ty_{ij})
    \right\}
    p(\lambda) \cdot \frac{1}{C(\bxi)},
\end{equation}
where $C(\bxi)$ is the normalizing constant.
For simplicity, we suppress the conditioning variables $\by, \bty, \bbeta,$
and $\rho$ in \eqref{eq:supp-poster-crm} and
\eqref{eq:supp-poster-crm-new}, writing the two distributions as
$\pi(\lambda)$ and $\Pi(\lambda \mid \bxi)$, respectively. In what follows,
the subscript $a$ denotes augmented quantities, and the subscript $\star$
denotes proposed quantities.
We augment the target with an additional random measure $\lama$ whose
distribution is specified using the fixed reference value
$\wbxi = \{\widetilde{\theta}_{ij}, \widetilde{h}_{ij}\}$. The augmented
target is therefore
\begin{equation} \label{eq:supp-crm-target}
    \pi(\lambda, \lama) \propto \pi(\lambda) \Pi(\lama \mid \wbxi)
\end{equation}
For the proposed random measures $\lams$ and $\lamsa$, the proposal
distribution is
\begin{equation} \label{eq:supp-crm-prop}
    \begin{aligned}
        q(\lams, \lamsa \mid \lambda, \lama) = \Pi(\lams \mid \bxi) \Pi(\lamsa \mid \bxi^\star),
    \end{aligned}
\end{equation}
where $\bxi^\star \equiv \bxi(\lams) = \{\theta_{ij}(\lams), h_{ij}(\lams)\}$.
Accordingly, define $\theta_{ij}^\star \equiv \theta_{ij}(\lams)$ and
$h_{ij}^\star \equiv h_{ij}(\lams)$.

Combining the augmented target in \eqref{eq:supp-crm-target} with the proposal
distribution in \eqref{eq:supp-crm-prop}, the Metropolis--Hastings acceptance
ratio is obtained as follows. The third equality substitutes
\eqref{eq:supp-poster-crm} and \eqref{eq:supp-poster-crm-new}; the final
equality cancels the prior factors, CRM mass terms, and normalizing constants.
\begin{equation*}
\begin{aligned}
    r = & \frac{\pi(\lams, \lamsa)}
    {\pi(\lambda, \lama)} 
    \times
    \frac{q(\lambda, \lama \mid \lams, \lamsa)}
    {q(\lams, \lamsa \mid \lambda, \lama)} \\
    = & \frac{\pi(\lams) \cdot \Pi(\lamsa \mid \wbxi)}
    {\pi(\lambda) \cdot \Pi(\lama \mid \wbxi)} 
    \times
    \frac{\Pi(\lambda \mid \bxi^\star) \cdot \Pi(\lama \mid \bxi)}
    {\Pi(\lams \mid \bxi) \cdot \Pi(\lamsa \mid \bxi^\star)} \\
    = & \prod_i\prod_j \left\{
        \frac{\exp[\theta_{ij}^\star\ty_{ij}+h_{ij}^\star-b(\theta_{ij}^\star;\lams)]}
             {\exp[\theta_{ij}\ty_{ij}+h_{ij}-b(\theta_{ij};\lambda)]}
        \frac{\exp[\widetilde{\theta}_{ij}\ty_{ij}+\widetilde h_{ij}-b(\widetilde{\theta}_{ij};\lamsa)]}
             {\exp[\widetilde{\theta}_{ij}\ty_{ij}+\widetilde h_{ij}-b(\widetilde{\theta}_{ij};\lama)]}
        \right.\\
    & \qquad \quad\times \left.
        \frac{\exp[\theta_{ij}^\star\ty_{ij}+h_{ij}^\star-b(\theta_{ij}^\star;\lambda)]}
             {\exp[\theta_{ij}\ty_{ij}+h_{ij}-b(\theta_{ij};\lams)]}
        \frac{\exp[\theta_{ij}\ty_{ij}+h_{ij}-b(\theta_{ij};\lama)]}
             {\exp[\theta_{ij}^\star\ty_{ij}+h_{ij}^\star-b(\theta_{ij}^\star;\lamsa)]}
        \right\}\\
    = & \prod_i\prod_j \exp\!\left\{ 
        (\theta_{ij}^\star-\theta_{ij})\ty_{ij}
        +h_{ij}^\star-h_{ij} 
        -b(\theta_{ij}^\star;\lams) +b(\theta_{ij};\lambda)
        -b(\widetilde{\theta}_{ij};\lamsa)+b(\widetilde{\theta}_{ij};\lama)
        \right.\\
    & \qquad \quad \left.
        -b(\theta_{ij}^\star;\lambda)+b(\theta_{ij};\lams)
        -b(\theta_{ij};\lama) +b(\theta_{ij}^\star;\lamsa)
        \right\}.
\end{aligned}
\end{equation*}

The proposal is generated sequentially: first draw \(\lams\sim\Pi(\cdot\mid\bxi)\), 
then set \(\bxi^\star=\bxi(\lams)\) and 
draw \(\lamsa\sim\Pi(\cdot\mid\bxi^\star)\). Consequently, under the reverse 
move from \((\lams,\lamsa)\) to \((\lambda,\lama)\),
\[
q(\lambda,\lama\mid\lams,\lamsa)
=
\Pi(\lambda\mid\bxi^\star)\Pi(\lama\mid\bxi).
\]
The factors $\lambda(\cdot)$ and $p(\lambda)$ cancel between the target and
proposal ratios. The normalizing constants associated with the forward and reverse 
proposals cancel,
$C(\bxi)C(\bxi^\star)/\{C(\bxi^\star)C(\bxi)\}=1$,
while the reference constant $C(\wbxi)$ cancels between the two augmented
target densities.

\subsubsection{Algorithm}

One $\lambda$-update consists of the following steps:
\begin{enumerate}
    \item At the current state $\lambda$ and $\lama$, compute $\bxi = \bxi(\lambda)$.
    \item Generate a proposal:
    \begin{enumerate}
        \item Draw $\lams \sim \Pi(\cdot \mid \bxi)$ and compute
        $\bxi^\star = \bxi(\lams)$.
        \item Draw $\lamsa \sim \Pi(\cdot \mid \bxi^\star)$.
    \end{enumerate}
    \item Accept $\lams$ with probability $\min\{1, r\}$, where
    $r$ is evaluated using the expression above.
\end{enumerate}
In implementation, a convenient choice for $\wbxi$ is a fixed
reference value obtained from an initial pilot fit or from a post--burn--in
chain state. This reference stabilizes the ratio estimator.

\subsubsection{Random Measure Drawing}

The factor $e^{-b(\theta_{ij};\lambda)}$ of $\Pi(\cdot \mid \bxi)$ in
\eqref{eq:supp-poster-crm-new} contains the total mass
$\int_{\YY}\exp(\theta_{ij}\ty)\lambda(d\ty)$ of the exponentially tilted
measure. This normalizing mass couples the CRM jumps and complicates direct
simulation of $\lambda$.
\begin{proposition}[Auxiliary-variable augmentation]\label{prop:supp-u-post}
Following \citet{james2009posterior} and
\citet{catalano2024unified}, introduce an auxiliary variable $u_{ij}$ for
each $\ty_{ij}$. This augmentation removes the coupling and restores a
jump-wise factorization of the CRM update. The conditional distribution of
$\bu$ given $(\bty,\bxi)$ is
\begin{equation} \label{eq:supp-u-post}
    p(\bu \mid \bty, \bxi)
    \propto
    \exp\left\{
        -\int_{\YY}
        \log\left[
            1
            +
            \sum_i \sum_j
            u_{ij}
            \exp\left[
                \theta_{ij} t
                +
                h_{ij}
            \right]
        \right]
        G_N(dt)
    \right\},
\end{equation}
where $G_N = \alpha G_0 + \sum_\ell n_\ell^\star \delta_{z_\ell^\star}$,
$\{z_\ell^\star\}$ are the unique values of $\{\ty_{ij}\}$, with
multiplicities $\{n_\ell^\star\}$.
\end{proposition}

\begin{proof}
The proof follows the augmentation argument in Proposition 2 of
\citet{alam2025dpglm}.
\end{proof}

Thus, the random-measure update proceeds by first sampling $\bu$ from
\eqref{eq:supp-u-post} and then sampling $\lambda$ conditional on $\bu$.
\begin{proposition}[Conditional CRM draw]\label{prop:supp-crm-post}
Specifically,
\begin{equation} \label{eq:supp-crm-post}
    \lambda \mid \bu, \bty, \bxi \overset{\text{d}}{=} 
    \lambda^o + \sum_\ell J_\ell \delta_{z_\ell^\star}
\end{equation}
where
$\lambda^o \sim CRM(\nu^o)$ has the L\'evy intensity
\[
\nu^o(ds,dt) =
\frac{1}{s}
\exp\left\{
    -s\left(
        1 + \sum_i \sum_j
        u_{ij}\exp\left[
            \theta_{ij}t+h_{ij}
        \right]
    \right)
\right\}
\,ds\,\alpha G_0(dt).
\]
The random jump sizes $\{J_\ell\}$ are generated from
\[
p
\left( J_\ell = 
s \mid \bu, \boldsymbol{\theta},
z_\ell^\star, n_\ell^\star
\right)
\propto
s^{n_\ell^\star-1}
e^{-s\{1+\psi(z_\ell^\star;\bu,\boldsymbol{\theta})\}}
=
\operatorname{Gamma}
\left(
n_\ell^\star,
\psi(z_\ell^\star;\bu,\boldsymbol{\theta})+1
\right).
\]
where
\[
\psi(z_\ell^\star; \bu, \boldsymbol{\theta})
=
\sum_i \sum_j u_{ij} e^{\theta_{ij} z_\ell^\star + h_{ij}}.
\]
\end{proposition}

\begin{proof}
The result follows from the conditional CRM representation in Proposition 3 of
\citet{alam2025dpglm}, after conditioning on the auxiliary variables $\bu$.
The residual CRM has L\'evy intensity $\nu^o$, while each observed atom
$z_\ell^\star$ receives a Gamma-distributed jump with shape
$n_\ell^\star$ and rate $1+\psi(z_\ell^\star;\bu,\boldsymbol{\theta})$.
\end{proof}

Note that integrating out $\bu$ from $\Pi(\lambda \mid \bu, \bty, \bxi)$ from 
the definition of \eqref{eq:supp-crm-post} with respect to $p(\bu \mid \bty, \bxi)$
yields $\Pi(\lambda \mid \bxi)$ in \eqref{eq:supp-poster-crm-new}.
Hence, to draw the random measure, we first sample $\bu = \{u_{ij}\}$ from
\eqref{eq:supp-u-post} using component-wise MH updates with 10 sweeps,
and then draw the random measure from \eqref{eq:supp-crm-post}.

\section{Prior Predictive Comparison}
\label{sec:supp-prior-comparison}

We compared implications of the CDPGLM and beta mixed-model priors before
observing data. This is a comparison of prior-induced marginal variability and
repeated-visit dependence. Figure~\ref{fig:supp-prior-comparison} summarizes
these prior-predictive comparisons across the three target mean response
levels. For each CDPGLM draw, we generated
the normalized baseline measure $G=\lambda/\lambda(\YY)$, represented by a
truncated stick-breaking approximation to a Dirichlet process,
\[
  G = \sum_{k=1}^{20} w_k\delta_{a_k},\quad
  V_k\sim\mathrm{Beta}(1,\alpha),\quad
  w_k=V_k\prod_{r<k}(1-V_r),\quad
  a_k\sim G_0=\mathrm{Uniform}(0,1),\quad \alpha=1.
\]
The final retained stick receives the remaining mass, and the weights are
normalized after truncation. Here $G$ is the prior random distribution of
response locations; $\alpha$ controls its concentration around $G_0$.

To make the marginal-variance comparison at a specified mean $\mu$, the script
exponentially tilts the atom weights by a scalar $\theta$:
\[
  w_k(\theta)=\frac{w_k\exp(\theta a_k)}
    {\sum_{r=1}^{20}w_r\exp(\theta a_r)},\qquad
  \sum_{k=1}^{20}w_k(\theta)a_k=\mu.
\]
It solves the second equation numerically for $\theta$, then computes the
tilted distribution's variance
\[
  v_{\mathrm{CDP}}(\mu)=\sum_{k=1}^{20}w_k(\theta)(a_k-\mu)^2.
\]
The comparison uses target means $\mu\in\{0.25,0.50,0.75\}$. At each fixed
mean, the beta response distribution is parameterized as
$Y\mid\mu,\phi\sim\mathrm{Beta}(\mu\phi,(1-\mu)\phi)$, with conditional
variance $v_{\mathrm{Beta}}(\mu)=\mu(1-\mu)/(\phi+1)$. The first row of
Figure~\ref{fig:supp-prior-comparison} compares this beta variance with the
corresponding CDPGLM variance at each target mean.

The bottom row displays the
prior-predictive Pearson correlation between two repeated outcomes,
$\operatorname{Corr}(Y_{i1},Y_{i2})$, at fixed mean response levels
$\mu\in\{0.25,0.50,0.75\}$. For the CDPGLM, this is the observed-scale
correlation induced by the tilted CRM distribution and Gaussian copula; for
the beta mixed model, it is induced by the subject-specific random intercept
and integrated over its distribution.

\begin{figure}[!htp]
\centering
\includegraphics[width=\textwidth]{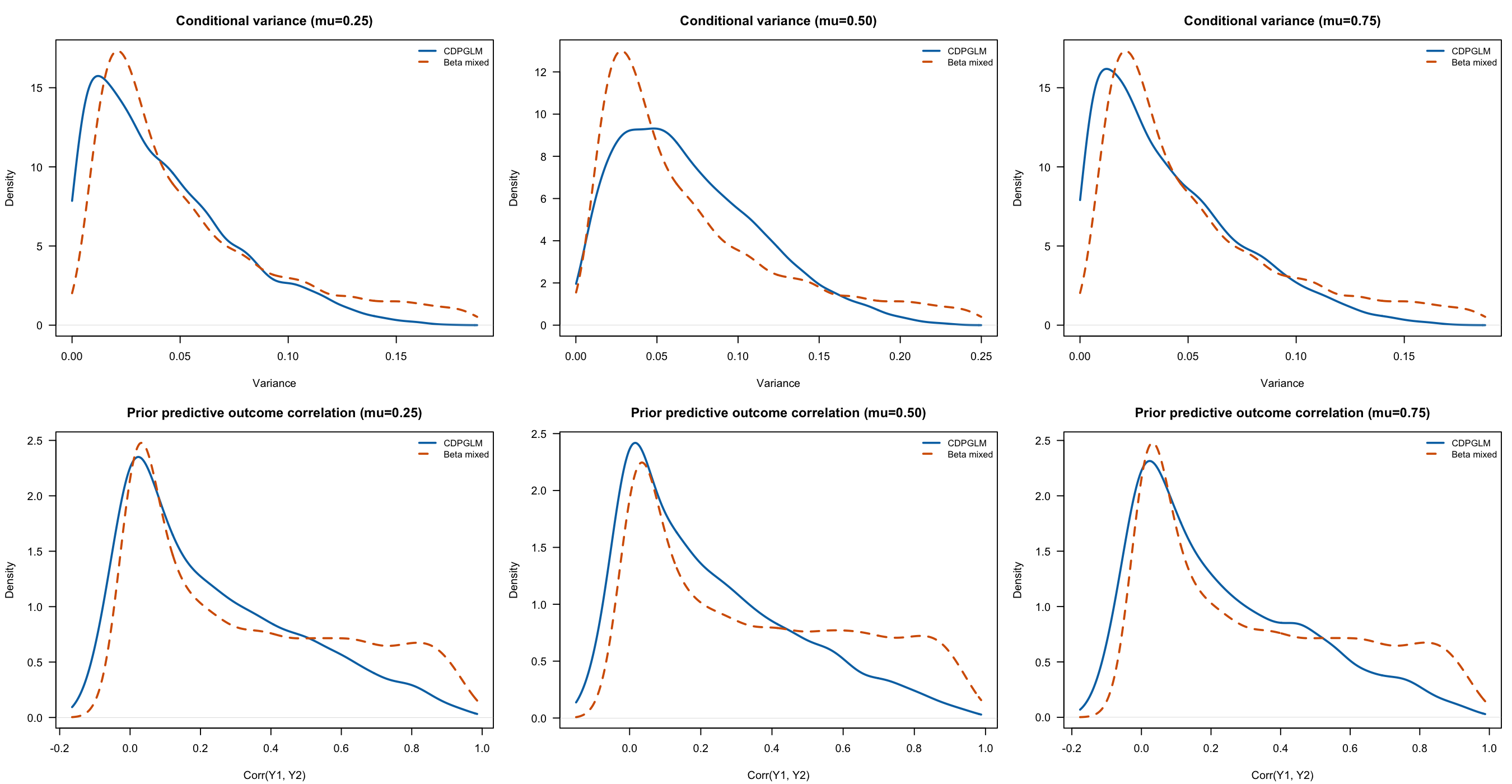}
\caption{Prior predictive comparison of the CDPGLM and beta mixed model. The
top row shows conditional response variance, and the bottom row shows
prior-predictive Pearson correlations between two repeated outcomes. Columns
correspond to target means $\mu=0.25$, $0.50$, and $0.75$.}
\label{fig:supp-prior-comparison}
\end{figure}

\section{Expanded Simulation Details and Results}
\label{sec:supp-sim-settings}

\subsection{Data--generating setting}

For the ET generator and simulation setting discussed in Section \ref{sec:simulation},
we used an exchangeable Gaussian copula with $\rho = 0.80$ and the empirical 
distribution $f_n(\by)$ of the LHS data, shown in Figure~\ref{fig:et_empirical_grid}, 
as the simulation truth for the normalized baseline distribution 
$\lambda(\cdot) / \lambda(\YY)$.

\begin{figure}[!htp]
\centering
\includegraphics[width=0.8\textwidth]{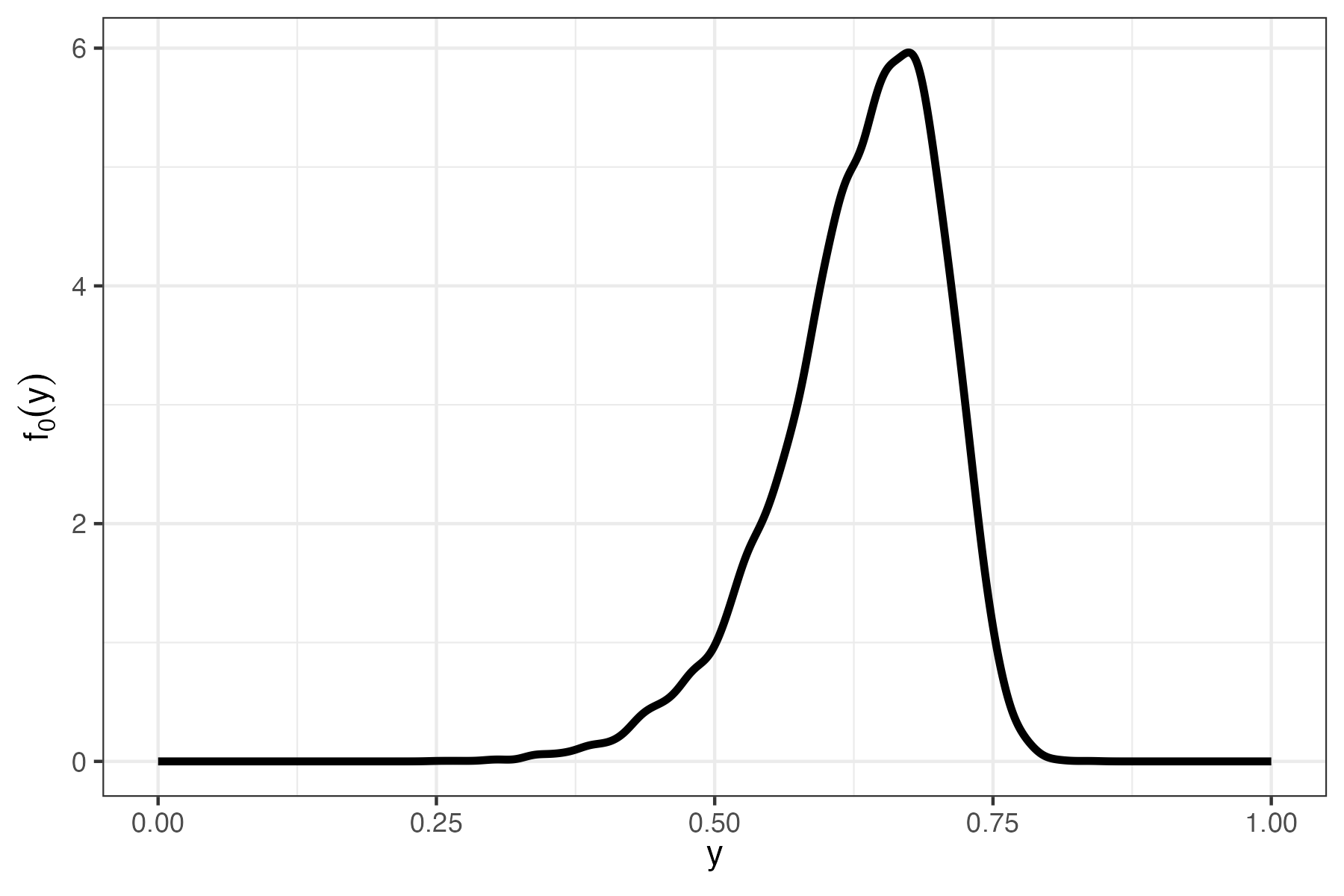}
\caption{Empirical distribution $f_n(\cdot)$ of the LHS data.}
\label{fig:et_empirical_grid}
\end{figure}

\subsection{Supplementary Regression Results}

Supplementary Figure~\ref{fig:sim_et_bias_sd} presents raw regression bias and
standard-deviation summaries by coefficient, method, and sample size. The
within-patient response-scale correlations are shown in Supplementary Figure~
\ref{fig:sim_et_within_patient_correlations}, and Supplementary Figure~
\ref{fig:sim_et_regression_rmse} reports fixed--effect RMSE by coefficient,
method, and sample size under the ET setting.

\begin{figure}[!htp]
\centering
\includegraphics[width=0.95\textwidth]{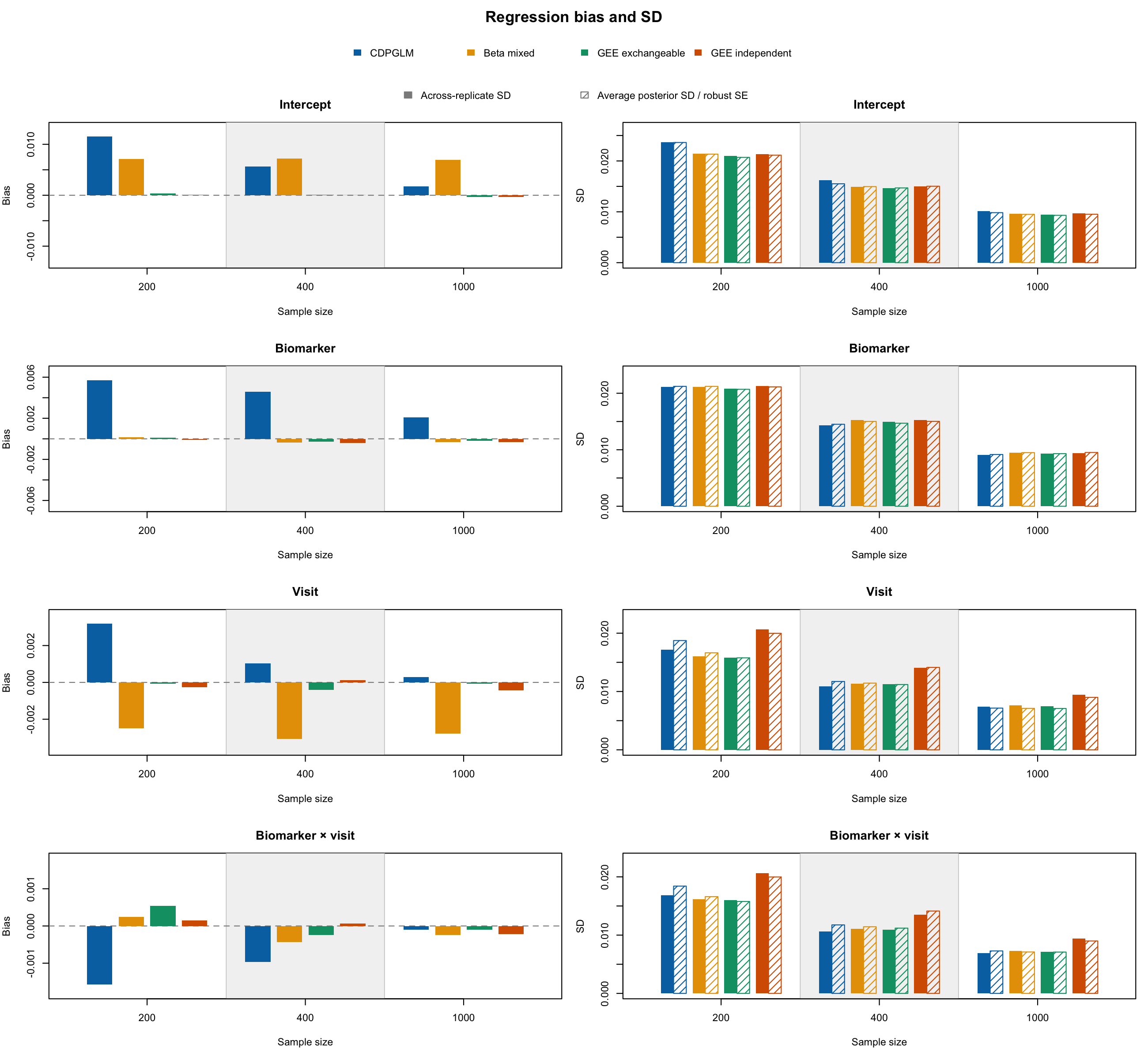}
\caption{ET simulation regression bias and standard deviation summaries. The
left column shows bias, and the right column shows across-replicate standard
deviation as solid bars together with average posterior standard deviation for
the Bayesian methods or average robust standard error for the GEE methods as
hatched bars. Rows correspond to regression coefficients, the horizontal axis
shows sample size, and colors identify methods. The beta mixed-model results
are shown for context; its coefficients are subject-specific and are not a
direct benchmark for the marginal coefficients of the CDPGLM and GEE methods.}
\label{fig:sim_et_bias_sd}
\end{figure}

\begin{figure}[!htp]
\centering
\includegraphics[width=\textwidth]{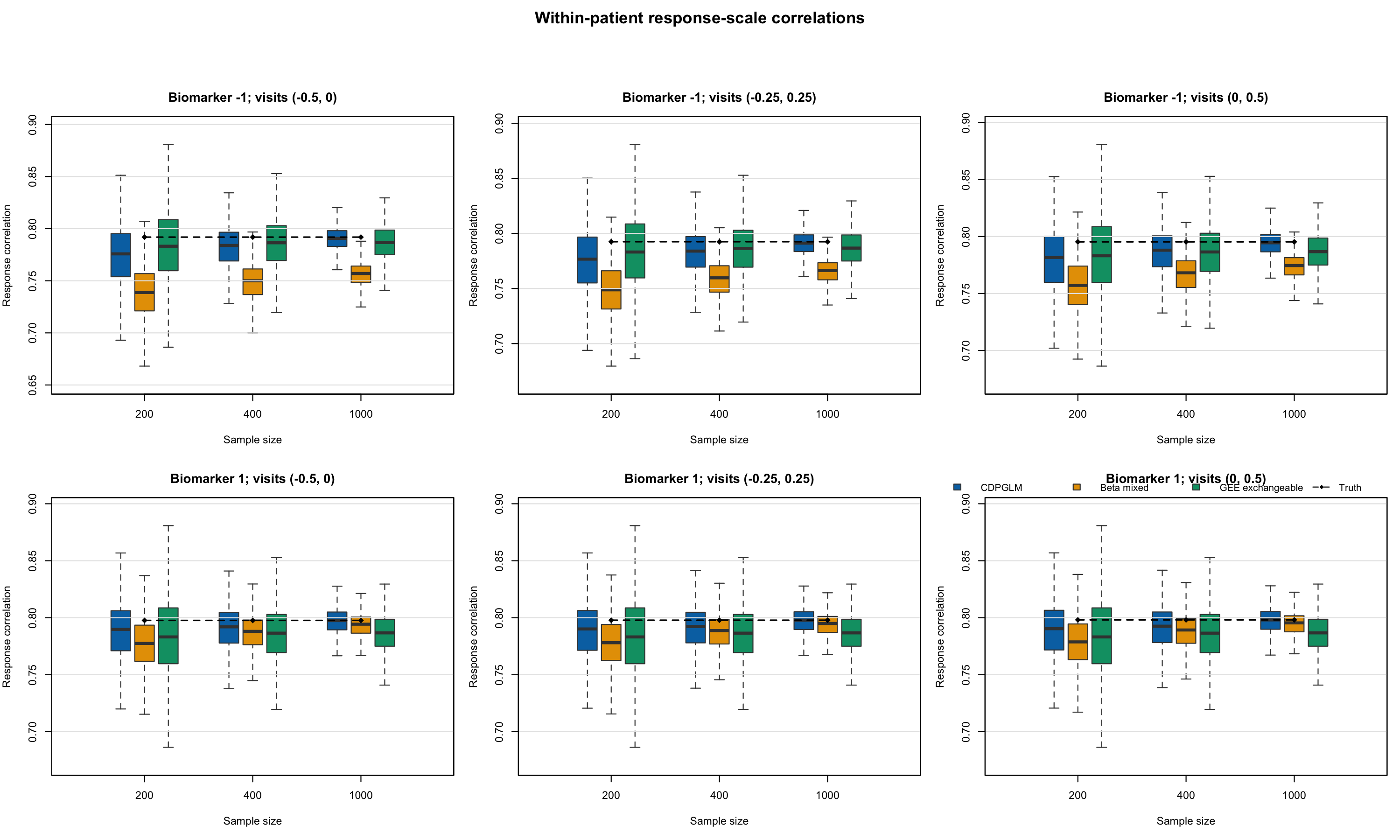}
\caption{ET simulation within-patient response-scale correlations at
representative biomarker profiles and visit pairs, by sample size and method.
Boxplots show the estimated correlations across simulation replicates, and
the dashed lines with points indicate the corresponding true correlations.}
\label{fig:sim_et_within_patient_correlations}
\end{figure}

\begin{figure}[!htp]
\centering
\includegraphics[width=\textwidth]{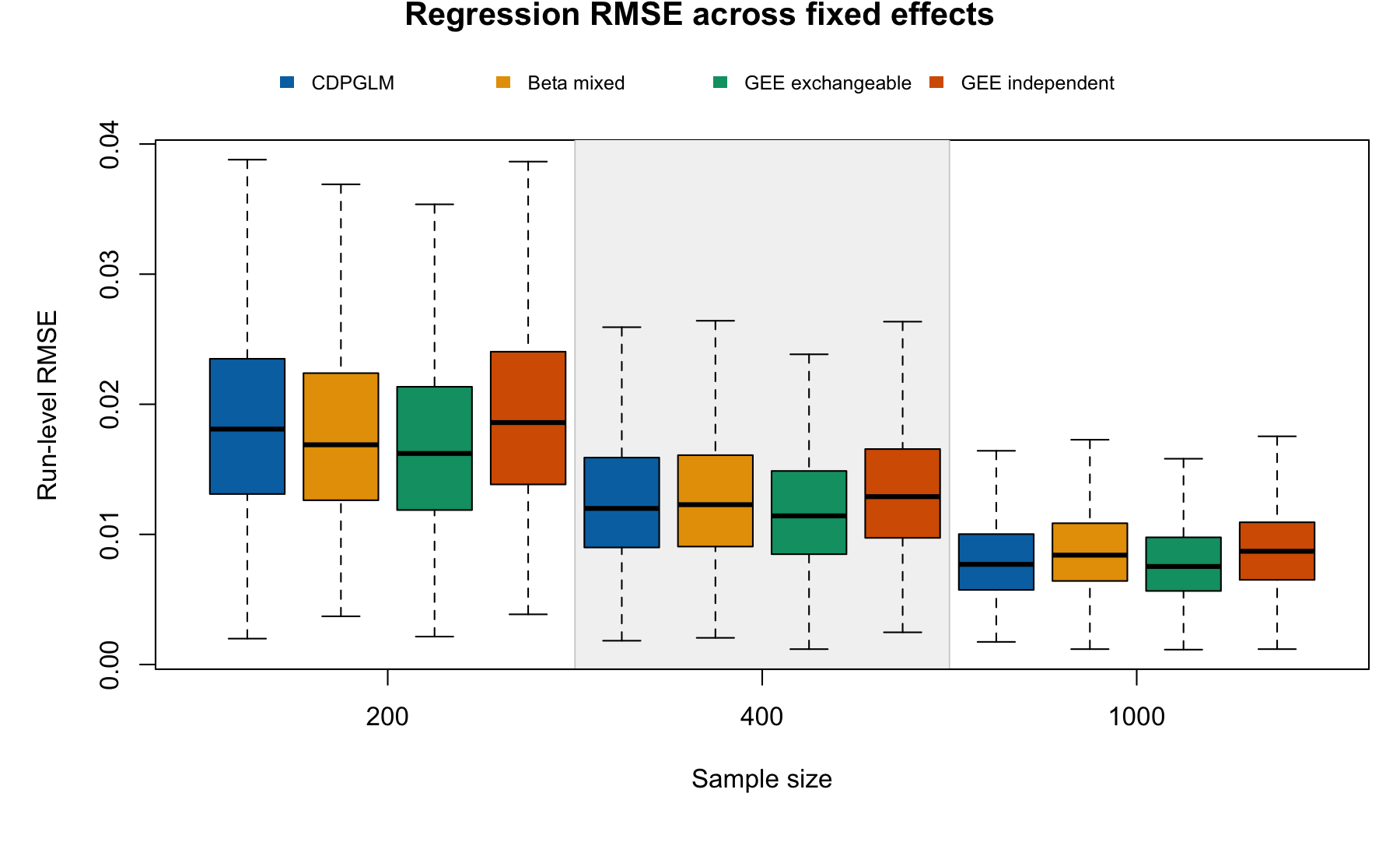}
\caption{ET simulation: regression RMSE across fixed effects and methods by
sample size. Lower values are better.}
\label{fig:sim_et_regression_rmse}
\end{figure}

\clearpage\newpage
\section{Additional Case Study Data Summary and Results}
\label{sec:supp-case-study}

The Lung Health Study analysis included 11,196 observations from 2,562 patients.
Patients contributed a mean of 4.370 observations (SD 1.002), with a median of
5 observations and a range of 2--5 observations.
The study population consisted of continuing smokers with available genotype
information and complete outcome and smoking-exposure covariates, and each
patient contributed at least two observed visits. Descriptive summaries for 
the covariates are reported in 
Tables~\ref{tab:supp-case-study-covariates}--\ref{tab:supp-case-study-site}.

\begin{table}[!htp]
\centering
\caption{Lung Health Study: descriptive summaries for continuous and binary
covariates. Continuous variables are summarized by their mean and standard
deviation; binary variables are summarized by the percentage at level 1.}
\label{tab:supp-case-study-covariates}
\begin{tabular}{lc}
\hline
Covariate & Mean (SD) or Percentage \\
\hline
Age & $-0.178$ ($0.670$) \\
BMI & $0.157$ ($0.796$) \\
Pack-years & $0.029$ ($0.920$) \\
Biomarker & 56.0\% \\
Female & 37.1\% \\
\hline
\end{tabular}
\par\smallskip
\noindent\begin{minipage}{0.92\textwidth}
\footnotesize\textit{Note.} Age, BMI, and pack-years were centered
and scaled for analysis. Age was centered at 50 and scaled by 10 units.
BMI was centered at 25 and scaled by 5 units.
Pack-years was centered at 40 and scaled by 20 units.
Biomarker and Female are binary covariates encoded as $-1$ and $1$; the
percentages reported in the table correspond to level $1$.
\end{minipage}
\end{table}

\begin{table}[!htp]
\centering
\caption{Lung Health Study: number of observations by visit.}
\label{tab:supp-case-study-visit}
\begin{tabular}{lrr}
\hline
Visit & Number of observations & Percentage \\
\hline
$-0.5$ & 2419 & 21.6\% \\
$-0.25$ & 2432 & 21.7\% \\
$0$ & 2232 & 19.9\% \\
$0.25$ & 2063 & 18.4\% \\
$0.5$ & 2050 & 18.3\% \\
\hline
\end{tabular}
\end{table}

\begin{table}[!htp]
\centering
\caption{Lung Health Study: number of patients by study site.}
\label{tab:supp-case-study-site}
\begin{tabular}{lrr}
\hline
Site & Number of patients & Percentage \\
\hline
1 & 230 & 9.0\% \\
2 & 258 & 10.1\% \\
3 & 242 & 9.4\% \\
4 & 261 & 10.2\% \\
5 & 302 & 11.8\% \\
6 & 254 & 9.9\% \\
7 & 246 & 9.6\% \\
8 & 310 & 12.1\% \\
9 & 276 & 10.8\% \\
10 & 183 & 7.1\% \\
\hline
\end{tabular}
\end{table}

The full posterior summaries
for the regression coefficients used in the analysis are reported in
Table~\ref{tab:supp-case-study-coefficients}. Entries are posterior mean
(posterior standard deviation), rounded to three decimal places.

\begin{table}[!htp]
\centering
\caption{Lung Health Study: posterior summaries for all regression coefficients
under the two alternative models. Entries are posterior mean (posterior
standard deviation), rounded to three decimal places.}
\label{tab:supp-case-study-coefficients}
\setlength{\tabcolsep}{5pt}
\begin{tabular}{lcc}
\hline
Parameter & Beta mixed & CDPGLM \\
\hline
Intercept & 0.463 (0.018) & 0.476 (0.017) \\
Biomarker & -0.007 (0.005) & -0.005 (0.005) \\
Visit & -0.129 (0.003) & -0.122 (0.003) \\
Age & -0.094 (0.008) & -0.089 (0.007) \\
Female & 0.012 (0.006) & 0.008 (0.008) \\
BMI & -0.001 (0.004) & -0.010 (0.005) \\
Pack-years & -0.024 (0.006) & -0.019 (0.008) \\
Site 2 & 0.079 (0.025) & 0.077 (0.025) \\
Site 3 & 0.083 (0.025) & 0.071 (0.028) \\
Site 4 & 0.065 (0.027) & 0.083 (0.018) \\
Site 5 & 0.015 (0.023) & 0.008 (0.017) \\
Site 6 & 0.113 (0.022) & 0.092 (0.026) \\
Site 7 & 0.014 (0.025) & 0.020 (0.019) \\
Site 8 & 0.121 (0.024) & 0.031 (0.019) \\
Site 9 & 0.127 (0.023) & 0.062 (0.020) \\
Site 10 & 0.039 (0.025) & 0.049 (0.023) \\
Biomarker $\times$ Visit & -0.009 (0.003) & -0.010 (0.003) \\
\hline
\end{tabular}
\end{table}

\end{document}